\documentclass[a4paper,10pt]{amsart}
\usepackage[T1]{fontenc}
\usepackage[latin1]{inputenc}
\usepackage{amssymb,amsmath,amsthm,amsrefs}

\newtheorem{definition}{Definition}
\newtheorem{theorem}{Theorem}
\newtheorem{lemma}{Lemma}

\newtheorem*{assp*}{\textbf{(P) Log-H\"older continuity condition}}
\newtheorem*{assi*}{\textbf{(I) Short-range interaction}}
\newtheorem*{dsknN*}{(\textbf{DS}.$k,n,N$)}
\numberwithin{equation}{section}
\numberwithin{theorem}{section}
\numberwithin{definition}{section}
\numberwithin{lemma}{section}
\DeclareMathOperator{\dist}{dist}
\DeclareMathOperator{\diam}{diam}

\DeclareMathOperator{\prob}{\mathbb{P}}
\newcommand{\condP}{\mathbf{(P)}}
\newcommand{\condI}{\mathbf{(I)}}
\newcommand{\PI}{\mathrm{PI}}
\newcommand{\FI}{\mathrm{FI}}
\newcommand{\ee}{\mathrm{e}}

\newcommand{\card}{\mathrm{card}}
\newcommand{\Bone}{\mathbf{1}}

\newcommand{\BH}{\mathbf{H}}
\newcommand{\BDelta}{\mathbf{\Delta}}
\newcommand{\sep}{\mathrm{sep}}

\newcommand{\dskNN}{(\textbf{DS}$.k,N,N$)}
\newcommand{\dsknN}{(\textbf{DS}$.k,n,N$)}
\newcommand{\dskonN}{(\textbf{DS}$.k+1,n,N$)}
\newcommand{\dsknprimeN}{(\textbf{DS}$.k,n',N$)}
\newcommand{\dskunnprimeN}{(\textbf{DS}$.k-1,n',N$)}
\newcommand{\dskprimenN}{(\textbf{DS}$.k',n,N$)}
\newcommand{\BPsi}{\mathbf{\Psi}}

\newcommand{\DZ}{\mathbb{Z}}
\newcommand{\DR}{\mathbb{R}}
\newcommand{\DN}{\mathbb{N}}
\newcommand{\esm}{\mathbb{E}}
\newcommand{\DP}{\mathbb{P}}

\newcommand{\BA}{\mathbf{A}}
\newcommand{\BB}{\mathbf{B}}
\newcommand{\BC}{\mathbf{C}}
\newcommand{\BX}{\mathbf{X}}
\newcommand{\BK}{\mathbf{K}}
\newcommand{\BP}{\mathbf{P}}
\newcommand{\BG}{\mathbf{G}}
\newcommand{\BI}{\mathbf{I}}
\newcommand{\BU}{\mathbf{U}}
\newcommand{\BV}{\mathbf{V}}

\newcommand{\Bx}{\mathbf{x}}
\newcommand{\By}{\mathbf{y}}
\newcommand{\Bu}{\mathbf{u}}
\newcommand{\Bv}{\mathbf{v}}

\newcommand{\FB}{\mathfrak{B}}

\newcommand{\CJ}{\mathcal{J}}
\newcommand{\CN}{\mathcal{N}}
\newcommand{\CR}{\mathcal{R}}

\title[Localization for multi-particle Poisson models]{Localization for one-dimensional two-particle random Schr\"odinger operators with Poisson potential}
\author{Tr\'esor Ekanga}
\address{Universit\'e Paris Diderot 13 Rue Albert Einstein 75013 Paris France}
\email{ekanga@math.cnrs.fr}
\begin{document}
\begin{abstract}
We prove the complete spectral and the strong dynamical Anderson localization  in a two-particle random Schr\"odinger operators with the Poisson potential. The results apply with sufficiently  weak interaction between the particle system.
\end{abstract}
\maketitle

\section{Introduction, assumptions and the main results}

\subsection{Introduction}
In this work we consider a system of two-particle  Anderson  model with a Poisson potential in the continuous  one-dimensional space  and prove the localization results (Anderson spectral localization and the strong dynamical localization ) for a sufficiently weakly interacting particle system. The novelty of this problem is that for the Poisson potential, we have a lack of monotonicity in the random parameter. A property which was successfully used in proofs of localization  for Anderson-type models \cites{SW86,KS87,CH94}.

This difficulty was earlier overcome  in the works by Stolz \cites{S95,S55}. Recall that in \cite{PF92}, the authors studied the spectra of random operators and almost periodic operators. We can find  in the books by Carmona et al. \cites{CL90,C83} some materials on spectral theory of random Schr\"odinger operators for one and higher dimensional models. 

The theory of multi-particle models  such as two-particle Anderson models is relatively recent and constitute a new direction in the spectral theory of random schr\"odinger operators \cites{AW09,CS09}.

In our earlier work \cite{Eka16} in multi-particle Anderson models in one dimension, we prove the complete spectral and strong dynamical localization  for the weakly interacting multi-particle system. While the continuous  version of the work can be found in \cite{Eka16}. 

Let now discuss, on the structure of the paper: in the next Section, we present  the model and state the assumptions and the main results. Section \ref{sec:MSA.scheme} is devoted to the initial length scale estimates of the multi-particle multi-scale analysis. In Section \ref{sec:ILS} we prove the initial length scales estimates of the multi-scale analysis. In Section \ref{sec:MP.induction}, we prove the multi-scale induction step of the multi-scale analysis. In the last Section, Section \ref{sec:proof.results}, we prove the main results on spectral localization Theorem \ref{thm:main.result.exp.loc} and dynamical localization Theorem \ref{thm:main.result.dynamical.loc}.

\subsection{The model and the assumptions}
the two-particle one-dimensional  Anderson model with a Poisson random potential is given by the Schr\"odinger Hamiltonian
\[
\BH^{(2)}_h(\omega)=-\BDelta + V(\Bx,\omega)+h\BU(\Bx), \quad \Bx=(x_1,x_2)\in\DR^2,
\] 
acting on $L^2(\DR^2)$, where 
\[
V(\Bx,\omega)=\sum_i f(\Bx-X_i(\omega)),
\]
with $f\in L^2(\DR)$ and where $\{X_i(\omega)\}$ is a finite set of points $X_i(\omega)\in\DR$ so that $V$ is a random variable  relative  to some probability space $(\Omega,\FB,\DP)$ and acts on $L^2(\DR^2)$ as a multiplication operator by the function $V(\Bx)$. Also $\BU$ is the interaction potential between the two-particle and acts on $L^2(\DR^2)$ as a multiplication  operator by the function $\BU(\Bx)$.

Set $\Omega=\DR^{DZ^d}$ and $\FB=\bigotimes B(\DR)$  where $B(\DR)$ is the Borel sigma-algebra on $\DR$. Let $\mu$ be a probability measure on $\DR$ and define $\DP=\bigotimes_{\DZ^d}\mu$ on $\Omega$.

\begin{assp*}
The random potential $V: \DZ^d\times \Omega\rightarrow \DR$ is i.i.d. and the corresponding probability distribution function $F_V$ is log-H\"older continuous: More precisely
\begin{align*}
&s(F_V,\varepsilon) :=\sup_{a\in\DR}(F_V(a+\varepsilon)-F_V(a))\leq \frac{C}{|\ln \varepsilon|^{2A}}\\
& \text{ for some $C\in(0,\infty)$ and $A\in (\frac{3}{2}\times 4^N+9Nd,\infty)$}.
\end{align*}
Further, the single-site potential $f$ is non negative and compactly supported.
\end{assp*}

\begin{assi*}
The interaction potential $\BU$ is bounded and  there exists $r_0\in\DN$ such that 
 \[
\BU(x_1,x_2)=0 \quad \text{if $|x_1-x_2|\geq r_0$}
\]
\end{assi*}

\subsection{The results}

\begin{theorem}\label{thm:main.result.exp.loc}
Let $d=1$. Under assumption $\condI$ and $\condP$ there exists $h^*\in(0,\infty)$ such that for any $h\leq |h^*|$ the Hamiltonian $\BH^{(N)}_h$ with interaction of amplitude  $|h|$ exhibits complete Anderson localization, i.e., with $\DP$-probability one, the spectrum of $\BH^{(N)}_h$ is pure point and each eigenfunction $\BPsi$ is exponentially decaying fast at infinity:
\[
\|\chi_{\Bx}\cdot\BPsi\| \leq C\ee^{-c|\Bx|}
\]
for some positive constants $c, C$.
\end{theorem}

\begin{theorem}\label{thm:main.result.dynamical.loc}
Under assumptions $\condI$ and $\condP$ there exists $h^*, s^*\in(0,\infty)$ 
such that for any $h\leq|h|$ and any $s\leq|s^*|$ and any compact 
domain $\BK\subset \DR^{Nd}$ we have that the quantity
\[
\esm\left[\sup_{t\geq 0}\left\| |\BX|^{\frac{s}{2}}\ee^{-it\BH^{(N)}(\omega)}
\BP_I(\BH^{(N)}(\omega))\Bone_{\BK}\right\|_{L^2(\DR^{Nd})}\right]
\] 
is finite, where $(|\BX|\BPsi)(\Bx):=|\Bx|\BPsi(\Bx)$. $\BP_I(\BH^{(N)}(\omega
))$ is the spectral projection onto the interval $I$ and $\Bone_{\BK}$ is the 
characteristic  function of the set $\BK$.
\end{theorem}

\section{The multi-particle multi-scale analysis scheme}\label{sec:MSA.scheme}

\subsection{Geometric facts}

According to the general structure  of the multi-scale analysis, we work with \emph{rectangular } domains.  For $\Bu=(u_1,\ldots,u_n)\in\DZ^{nd}$, we denote by $\BC^{(n)}_L(\Bu)$ the $n$-particle cube, i.e.,
\[
\BC^{(n)}_L=\left\{\Bx\in\DR^{nd}: |\Bx-\Bu|\leq L\right\},
\]
and given $\{L_i: i=1,\ldots,n\}$, we define the rectangle
\begin{equation}
\BC^{(n)}(\Bu)=\prod_{i=1}^n C_{L_i}^{(1)}(u_i),
\end{equation}
where $C^{(1)}_{L_i}(u_i)$ are the cubes of side length $2L_i$, center at points $u_i\in\DZ^d$.   We also define 
\[
\BC^{(n,int)}_L(\Bu):=\BC^{(n)}_{L/3}(\Bu),\quad \BC^{(n,out)}_L(\Bu):=\BC^{(n)}_L(\Bu)\setminus\BC^{(n)}_{L-2}(\Bu), \quad \Bu\in\DZ^{nd}
\]
 and introduce the characteristic functions:
\[
\Bone_{\Bx}^{(n,int)}:=\Bone_{\BC^{(n,int)}_L(\Bx)}, \quad \Bone_{\Bx}^{(n,out)}:=\Bone_{\BC^{(n,out)}_L(\Bx)}.
\]
The volume of the cube $\BC^{(n)}_L(\Bu)$ is $|\BC^{(n)}_L(\Bu)|=(2L)^{nd}$. We denote the restriction of the Hamiltonian $\BH^{(n)}$ to $\BC^{(n)}(\Bu)$ by 

\begin{align*}
& \BH^{(n)}_{\BC^{(n)}(\Bu)}=\BH^{(n)}|_{\BC^{(n)}(\Bu)}\\
& \text{ with dirichlet boundary conditions}\\
\end{align*}
We denote the spectrum of $\BH^{(n)}_{\BC^{(n)}(\Bu)}$ by $\sigma(\BH^{(n)}_{\BC^{(n)}(\Bu)})$ and its resolvent by 
\[
\BG^{(n)}_{\BC^{(n)}(\Bu)}(E):=\left(\BH^{(n)}_{\BC^{(n)}(\Bu)}-E\right)^{-1}, \quad E\in\DR\setminus \sigma\left(\BH^{(n)}_{\BC^{(n)}(\Bu)}\right).
\]

Let $m$ be a positive constant and consider $E\in\DR$. A cube $\BC^{(n)}_L(\Bu)\subset\DR^{nd}$, $1\leq n\leq N$ will be called  $(E,m)$-non-singular ( $(E,m)$-NS) if $E\notin \sigma(\BH^{(n)}_{\BC^{(n)}_L(\Bu)})$ and 
\[
\| \Bone_{\Bx}^{(n,out)}\BG^{(n)}_{\BC^{(n)}_L(\Bx)}(E)\Bone_{\Bx}^{(n,int)}\|\leq \ee^{-\gamma(m,L,n)L},
\]
where 
\[
\gamma(m,L,n)=m(1+L^{-1/8})^{N-n+1}.
\]
Otherwise, it is called $(E,m)$-singular  ($(E,m)$-S).

Let us introduce the following:

\begin{definition}
Let $n\geq 1$, $E\in\DR$ and $\alpha=3/2$.
\begin{enumerate}
\item [A)]
A cube $\BC^{(n)}_L(\Bv)\subset \DR^{nd}$ is called $E$-resonant ($E$-R) if 
\[
\dist\left[E,\sigma(\BH^{(n)}_{\BC^{(n)}_L(\Bv)})\right]\leq \ee^{-L^{1/2}},
\]
otherwise, it is called $E$-non-resonant ($E$-R).

\item[B)] 
A cube $\BC^{(n)}_L(\Bv)\subset\DR^{nd}$ is called $E$-completely non-resonant ($E$-CNR), if it does not contain any $E$-R cube of size $\geq L^{1/\alpha}$. In particular $\BC^{(n)}_L(\Bv)$ is itself $E$-NR
\end{enumerate}
\end{definition}

We will also make use of the following notion,
\begin{definition}
A cube $\BC^{(n)}_L(\Bx)$  is $\CJ$-separable  from $\BC^{(n)}_L(\By)$ if there exists a non empty subset $\CJ\subset\{1,\ldots,n\}$ such that
\[
\left(\bigcup_{j\in\CJ} C^{(1)}_L(x_j)\right)\cap \left(\bigcup_{j\notin\CJ} C^{(1)}_L(x_j)\cup\bigcup_{j=1}^n C^{(1)}_L(y_j)\right)=\emptyset.
\]
A pair $(\BC^{(n)}_L(\Bx), \BC^{(n)}_L(\By))$ is separable if $|\Bx-\By|\geq 7NL$ and if one of the cube is $\CJ$-separable from the other. 
\end{definition}
\begin{lemma}\label{lem:separable}
Let $L\geq 1$.
\begin{enumerate}
\item[A)]
For any $\Bx\in\DZ^{nd}$, there exists a collection of $n$-particle cubes $\BC^{(n)}_{2nL}(\Bx^{(\ell)})$ with $\ell=1,\ldots,\kappa(n)$, $\kappa(n)=n^n$, $\Bx^{(\ell)}\in\DZ^{nd}$ such that if $\By\in\DZ^{nd}$ satisfies $|\By-\Bx|\geq 7NL$ and 
\[
\By\notin\bigcup_{\ell=1}^{\kappa(n)} \BC^{(n)}_{2nL}(\Bx^{(\ell)})
\]
then the cubes $\BC^{(n)}_L(\Bx)$ and $\BC^{(n)}_L(\By)$ are separable.
\item[B)]
Let $\BC^{(n)}_L(\By)\subset \DR^{nd}$ be an $n$-particle cube. Any cube $\BC^{(n)}_L(\Bx)$ with
\[
|\By-\Bx|\geq \max_{1\leq i,j\leq n}|y_i-y_j| + 5NL
\]
 is $\CJ$-separable from $\BC^{(n)}_L(\By)$ for some $\CJ\subset\{1,\ldots,n\}$.
\end{enumerate}
\end{lemma}
\begin{proof}
See the appendix Section \ref{sec:appendix}
\end{proof}     
\subsection{The multi-particle Wegner estimates}
In our earlier work \cite{Eka19a} as well in other previous papers in the multi-particle localization theory \cites{BCSS10b,CS09} the notion of separability was crucial in order to prove the Wegner estimates for pairs of multi-particle cubes via the Stollmann's Lemma. It is Plain (cf. \cite{Eka19a} Section 4.1) that sufficiently distant pairs of fully interactive cubes have disjoint projections and this fact combined with independence is used in that case to bound the probability of an intersection of events relative to those projections. We state below the Wegner estimates directly in a form suitable to the multi-particle multi-scale analysis using assumption $\condP$.

\begin{theorem}\label{thm:Wegner}
Assume that the random potential satisfies assumption $\condP$, then
\begin{enumerate}
\item[A)]
For any $E\in\DR$
\[
\prob\left\{\BC^{(n)}_L(\Bx) \text{ is not $E$-CNR}\right\} \leq L^{-p4^{N-n}}.
\]
\item[B)]
\[
\prob\left\{\exists E\in \DR \text{ neither $\BC^{(n)}_L(\Bx)$ nor $\BC^{(n)}_L(\By)$ is $E$-CNR}\right\}\leq L^{-p4^{N-n}}
\]
\end{enumerate}
where $p\geq 6Nd$, depends only on the fixed number of particles $N$ and the configuration dimension $d$.
\end{theorem}

\begin{proof}
See the articles \cites{BCSS10b,CS08}.
\end{proof}
We also give the Combes-Thomas estimates in 

\begin{theorem}\label{thm:CT}
Let $H=-\Delta + W$ be a Schr\"odinger operator on $L^2(\DR^D)$, $E\in\DR$ and $E_0=\inf\sigma(H)$. Set $\eta=\dist(E,\sigma(H))$. If $E$ is less than $E_0$, then for any $\gamma\in(0,1)$, wee have that
\[
\|\Bone_{\Bx}(H-E)^{-1}\Bone_{\By}\|\leq \frac{1}{(1-\gamma^2)\eta}\ee^{\gamma\sqrt{\eta d}}\ee^{-\gamma\sqrt{\eta}|\Bx-\By|},
\]
for all $\Bx,\By\in\DR^{D}$
\end{theorem}
\begin{proof}
See the proof of Theorem 1 in \cite{GK02}.
\end{proof}

We define the mass $m$ depending on the parameters $N$, $\gamma$ and the initial length scale $L$ in the following way:
\[
m:=\frac{2^{-N}\gamma L^{-1/4}}{3\sqrt{2}}.
\]
We recall below the geometric resolvent and the eigenfunction decay inequalities.
\begin{theorem}[Geometric resolvent inequality (GRI)]\label{thm:GRI}
For a given bounded $I_0\subset\DR$. There is a positive constant $C_{geom}$ such that for $\BC_{\ell}^{(n)}(\Bx)\subset\BC^{(n)}_L(\Bu)$, $\BA\subset\BC^{(n,int)}_{\ell}(\Bx)$, $\BB\subset\BC^{(n)}_L(\Bu)\setminus\BC^{(n)}_{\ell}(\Bx)$ and $E\in I_0$, the following inequality holds true:

\begin{gather*}
\|\Bone_{\BB}\BG^{(n)}_{\BC^{(n)}_L(\Bu)}(E)\Bone_{\BA}\|\leq C_{geom}\cdot\|\Bone_{\BB}\BG^{(n)}_{\BC^{(n)}_L(\Bu)}(E)\Bone_{\BC^{(n,int)}_{\ell}(\Bx)}\|\cdot\\ 
\|\Bone_{\BC^{(n,out)}_{\ell}(\Bx)}\|\cdot\|\Bone_{\BC^{(n,out)}_{\ell}(\Bx)}\BG^{(n)}_{\BC^{(n)}_{\ell}(\Bx)}(E)\Bone_{\BA}\|.
\end{gather*}
\end{theorem}
\begin{proof}
See \cite{Sto01}, Lemma 2.5.4.
\end{proof}
 
\begin{theorem}[Eigenfunctions decay inequality (EDI)]\label{thm:EDI}
For every $E\in\DR$, $\BC^{(n)}_{\ell}(\Bx)\subset\DR^{nd}$ and every polynomially bounded function $\BPsi\in L^2(\DR^{nd})$:
\[
\|\Bone_{\BC^{(n)}_1(\Bx)}\cdot\BPsi\|\leq C\cdot\|\Bone_{\BC^{(n,out)}_{\ell}(\Bx)}\BG^{(n)}_{\BC^{(n)}_{\ell}(\Bx)}(E)\Bone_{\BC^{(n,int)}_{\ell}(\Bx)}\|\cdot\|\Bone_{\BC^{(n,out)}_{\ell}(\Bx)}\cdot \BPsi\|.
\]

\end{theorem}
\begin{proof}
See Section 2.5 and Proposition 3.3.1. in \cite{Sto01}.
\end{proof}

\section{The initial bounds of the multi-particle multi-scale analysis}\label{sec:ILS}

\subsection{The fixed energy MSA bound for the n-particle system without interaction}

We begin with the well known single-particle exponential localization for the eigenfunctions and for one-dimensional Anderson models in the continuum proved in the paper by Damanik et al. \cite{DSS02}. Let $H^{(1)}_{C^{(1)}_L(x)}(\omega)$ be the restriction of the single-particle Hamiltonian into the cube $C^{(1)}_L(x)$ and denote by $\{\lambda_j,\phi_j\}_{j\geq 0}$, its eigenvalues and corresponding eigenfunctions. We have the following namely the single-particle exponential localization for the eigenfunctions in any cube.

\begin{theorem}[Single-particle localization]
There exists a constant $\tilde{\mu}\in(0,\infty)$ such that for every generalized eigenfunctions $\varphi$ of the single-particle Hamiltonian $H^{(1)}_{C^{(1)}_L(x)}(\omega)$ we have:
\[
\esm\left[\left\| \Bone_{C^{(1,out)}_L(x)}\cdot\varphi\cdot\Bone_{C^{(1,int)}_L(x)}\right\|\right]\leq \ee^{-\tilde{\mu}L}.
\]
\end{theorem}
\begin{proof}
We refer to the book by Stollmann \cite{Sto01}.
\end{proof}
The main result of this subsection is Theorem \ref{thm:ILS.np} given below. The proof of Theorem \ref{thm:ILS.np} relies on an auxiliary statement Lemma \ref{lem:NR.NS}. We need to introduce first $\{(\lambda_{j_i}^{(i)}, \varphi^{(i)}_{j_i}): j_i\geq 1\}$ the eigenvalues and the corresponding eigenfunctions of $H^{(1)}_{C^{(1)}_L(x_i)}(\omega)$, $i=1,\ldots,n$. Then the  eigenvalues of the non-interacting multi-particle random Hamiltonians $\BH^{(n)}_{\BC^{(n)}_L(\Bu)}(\omega)$ are written as sums:
\[
E_{j_1\cdots j_n} =\sum_{j=1}^n \lambda_{j_i}^{(i)}=\lambda^{(1)}_{j_1}+\cdots+\lambda^{(n)}_{j_n}
\]
while the corresponding eigenfunctions $\BPsi_{j_1\cdots j_n}$ can be chosen as tensor products
\[
\BPsi_{j_1\cdots j_n}=\psi^{(1)}_{j_1}\otimes \cdots\otimes \psi^{(n)}_{j_n},
\]
The eigenfunctions of finite volume Hamiltonians are assumed normalized.

\begin{theorem}\label{thm:ILS.np}
Let $1\leq n\leq n$ and $I_0\subset\DR$ a bounded interval. there exists $m^*\in(0,\infty)$ such that for any cube $\BC^{(n)}_L(\Bu)$ and all $E\in I_0$,
\[
\prob\left\{\BC^{(n)}_L(\Bu) \text{ is $(E,m^*,0)$-S}\right\}\leq \frac{1}{2} L_0^{-2p^*4^{N-n}}
\]
with $L_0$ large enough and $p^*\in(6Nd,\infty)$.
\end{theorem}
The proof of Theorem \ref{thm:ILS.np} relies on the following auxiliary statement.
\begin{lemma}\label{lem:NR.NS.gamma}
Let be given $N\geq n\geq 2$ $m^*\in(0,\infty)$ a cube $\BC^{(n)}_L(\Bu)$ and $E\in\DR$. Suppose that $\BC^{(n)}_L(\Bu)$ is $E$-NR, and for any operator $H^{(1)}_{C^{(1)}_L(x)}$ all its eigenfunctions $\psi_{j_i}$ satisfy 
\[
\Bone_{C^{(1)}_{L_0}(u_i)}\cdot\psi_{j_i}\cdot\Bone_{C^{(1,int)}_{L_0}(u_i)}\|\leq \ee^{-2\gamma(m^*,L_0,n)L_0}
\]
Then $\BC^{(n)}_{L_0}(\Bu)$ is $(E,m^*,L_0)$-NS provided that $L_0\geq L_*(m^*,N,d)$
\end{lemma}  
\begin{proof}
We choose the multi-particle eigenfunctions as tensor products of those of the single-particle Hamiltonians $H^{(1)}_{C^{(1)}_L{0}(u_i)}(\omega)$, $i=1,\ldots,n$, i.e., $\BPsi_j=\varphi_j^{(1)}\otimes\cdots\otimes \varphi_j^{(n)}$ corresponding to the eigenvalues $E_j=\lambda_j^{(1)}+\cdots+\lambda^{(1)}_j$. Now we have that,
\[
\BG^{(n)}_{\BC_{L_0}^{(n)}(\Bu)}(E)=_sum_{E_j\in\sigma(\BH^{(n)}_{\BC^{(n)}_{L_0}(\Bu)})}\BP_{\varphi^{(1)}_j}\otimes\cdots\otimes\BP_{\varphi_j^{(n-1)}} G^{(1)}_{C^{(1)}_{L_0}(u_n)}(E-\lambda_{\neq_j})
\]
where $\lambda_{\neq n}=\sum_{1\leq i\leq n-1}\lambda_i$ so that
\begin{gather*}
\Bone_{\BC^{(n,out)}_{L_0}(\Bu)}\BG^{(n)}_{\BC^{(n)}_{L_0}(\Bu)}(E)\Bone_{\BC^{(n,int)}_{L_0}(\Bu)}\leq \Bone_{\BC^{(n)}_L(\Bu)}\BG^{(n)}_{\BC^{(n)}_{L}(\Bu)}(E)\Bone_{\BC^{(n)}_L(\Bu)}\\
\leq \sum_{i=1}^n\Bone^{(\otimes(i-1)}\otimes\Bone_{C^{(1)}_L(u_i)}\otimes\Bone^{\otimes(n-i)}\BG^{(n)}_{\BC^{(n)}_L(\Bu)}(E)\\
\leq \sum_{i=1}^n\left[\sum_j \Bone^{\otimes(i-1)}\otimes\Bone^{(n-i)}\BP_{\varphi_j^{(1)}}\otimes\dots\otimes \BP_{\varphi_j^{(n-1)}} G^{(1)}_{C^{(1)}_L(u_n)}(E-\lambda_{\neq n})\right]
\end{gather*}  

By the Weyl's law there exists $E^*\in(0,\infty)$ such that $\lambda_j\geq E^*$ for all $j\geq j^*=C_{Weyl}|C^{(1)}_{L_0}(u_n)|$. Therefore, we divide the above sum into two parts as follows
\begin{gather*}
\Bone_{\BC^{(n,out)}_{L_0}(\Bu)}\BG^{(n)}_{\BC^{(n)}_{L_0}(\Bu)}(E)\Bone_{\BC^{(n,int)}_{L_0}(\Bu)}\leq \sum_{i=1}^n\left(\sum_{j\leq j^*}+\sum_{j\geq j^*+1}\right)\\
\times \Bone^{\otimes(i-1)}\otimes \Bone_{C^{(1)}_{L_0}(u_i)}\otimes\Bone^{\otimes(n-i)}\BP_{\varphi_j^{(1)}}\otimes\cdots\otimes\BP_{\varphi_j^{(n-1)}}G^{(1)}_{C^{(1)}_{L_0}(u_n)}(E-\lambda_{\neq})
\end{gather*}
Since 
\begin{gather*}
\| \Bone_{C^{(1)}_{L_0}(u_i)}\otimes\Bone^{\otimes(n-i)}\BP_{\varphi_j^{(1)}}\otimes\cdots\otimes\BP_{\varphi_j^{(n-1)}}G^{(1)}_{C^{(1)}_{L_0}(u_n)}(E-\lambda_{\neq})\|\\
\leq \|\Bone_{C^{(1)}_{L}(u_i)}\cdot\varphi_j^{(1)}\|\cdot\ee^{L^{1/2}}\leq \ee^{-2\gamma(m^*,l,1)L+L^{1/2}}
\end{gather*}   
for $L\in(L^*(N,d,C_{Weyl},\infty)$ large enough where we used the hypotheses on the exponential decay of the eigenfunctions of the single-particle Hamiltonian. Thus, the infinite sum can be made as small as an exponential decay provided that the length $L_0$ is large enough,

\[
\sum_{j\geq j^*+1} \|\Bone^{\otimes(i-1)}\otimes \Bone_{C^{(1)}_L(u_i)}\otimes \Bone^{\otimes(n-i)}\BP_{\varphi_j^{(1)}}\otimes \cdots\otimes\BPsi_{\varphi_j^{(n-1)}}G^{(1)}_{C^{(1)}_L(u_i)}(E-\lambda_{\neq n})
\|\leq \frac{1}{2} \ee^{-\gamma(m^*,L_0,n)L_0}
\]
while the finite can be bounded by:
\[
n\cdot C_{Weyl}\cdot |C^{(1)}_L(u)|\cdot\ee^{-\gamma(m^*,n,L)L}\ee^{L^{1/2}}\leq \frac{1}{2}\ee^{-\gamma(m,L,n)L},
\]
which proves the Lemma.

\end{proof} 
 Now we turn to the proof of Theorem \ref{thm:ILS}.

\begin{proof}[Proof of Theorem \ref{thm:ILS}]
Recall that by the single-particle Anderson localization theory there exists $\tilde{\mu}\in(0,\infty)$ such that such that we have the following decay bound on the exponential decay of the eigenfunctions: for all $\Bu\in\DZ^d$,
\begin{equation}\label{eq:mu}
\|\Bone_{C^{(1)}_L(u)}\cdot\Psi\|\leq \ee^{-\tilde{\mu}|x|}.
\end{equation}
Set $m^*=2^{-N-1}\tilde{\mu}$ and introduce the events 
\begin{align*}
&\CN:=\{\exists i=1,\ldots,n: \exists \lambda_j\in\sigma(H^{(1)}_{C^{(1)}_L(u_i)}(\omega)):\|\Bone_{C^{(1)}_L(u_i)}\cdot\varphi_j(u_i)\|\geq \ee^{-2\gamma(m^*,L_0,n)L_0},\\
& \CR:=\{\BC^{(n)}_L(\Bu) \text{ is $E$-NR}\}
\end{align*}
Then by Lemma \ref{lem:NR.NS.gamma}, Eqn  \eqref{eq:mu} and Theorem \ref{thm:Wegner} (A), we have:
\begin{align*}
\prob\left\{ \text{$\BC^{(n)}_L(\Bu)$ is $(E,m^*,0)$-S}\right\}&\leq \prob\{\CN\}+\prob\{\CR\}\\
& \leq \sum_{i=1}^n\sum_{j\geq 0} \frac{\esm\left[\|\Bone_{C^{(1)}_{L_0}(u_i)}\cdot\varphi_j\|\right]}{\ee^{-2\gamma(m^*,L_0,n)L_0}}+ \prob\{\CR\}\\
& \sum_{i=1}^n \sum_{j\geq 0} \ee^{(-\tilde{mu}_1+2\gamma(m^*,L_0,n)L_0)}+L_0^{-4^Np}.
\end{align*}
Since $2\gamma(m^*,n,L)\leq 2^{N+1}m^*\leq \tilde{\mu_1}$, $\tilde{\mu_1}-2\gamma(m,L,n)\in(0,\infty)$. Using the Weyl's law, we can divide the infinite sum above into two sums. Namely there exists a positive $E^*$ arbitrarily large such that $\lambda_j^{(1)}\geq E^*$ for $j\geq j^*=C^{(1)}_L(u_i)$ which yields 
\[
\sum_{j\geq 0}\ee^{(-\tilde{\mu_1}+2\gamma(m,n,L_0))L_0}=\left(\sum_{j\leq j^*}+\sum_{j\geq j^*+1}\right) \ee^{(-\tilde{\mu_1}+2\gamma(m,n,L_0))L_0}.
\].               
Above, the infinite sum can be made small than any polynomial power law provided that $L_0$ is large enough. We have 
\[
\sum_{i} \left(\sum_{j\leq j^*}+\sum_{j\geq j^*+1}\right) \ee^{(-\tilde{\mu_1}+2\gamma(m^*,L_0,n)L_0}\leq \frac{1}{3} L_0^{-2p4^{N-n}}+\frac{1}{3} L_0^{-2p4^{N-n}}+ L_0^{-p4^N}\leq L_0^{-2p4^{N-n}}.
\]
\end{proof}

We state and give here the proof of some important results from the paper \cite{Eka16} which use the fact that we are in the weakly interacting regime. The positive constant $m^*$ is the one from theorem \ref{thm:ILS}

\subsection{The fixed energy MSA bound for weakly interacting multi-particle systems}
Now we derive the required initial estimate from its counterparts established for non-interacting  systems.

\begin{theorem}\label{thm:ILS.WI}
$1\leq n\leq N$. Suppose that the Hamiltonians $\BH^{(n)}_0(\omega)$ (without inter-particle interaction) fulfills the  following condition for all $E\in I$ and all $\Bu\in\DZ^{nd}$

\begin{equation}\label{eq:ILS.NI}
\prob\left\{ \text{$\BC^{(n)}_{L_0}(\Bu)$ is $(E,m^*,0)$-S}\right\}\leq \frac{1}{2} L_0^{-2p^*4^{N-n}}
\qquad \text{with $p^*\in(6Nd,+\infty)$}.
\end{equation}
Then there exists $h^*\in(0,\infty)$ such that for all $h\in(-h^*,h^*)$ the Hamiltonian $\BH^{(n)}_h(\omega)$ with interaction of amplitude $|h|$  satisfies a similar bound. There exist some $p\in(6Nd,+\infty)$, $m\in(0,+\infty)$ such that for all $E\in I$ and all $\Bu\in\DZ^{nd}$
\[
\prob\{\text{$\BC^{(n)}_{L_0}(\Bu)$ is $(E,m,h)$-S}\}\leq \frac{1}{2} L_0^{-2p4^{N-n}}
\]
\end{theorem} 
\begin{proof}
First observe that the result of \eqref{eq:ILS.NI} is proved in the statement of Theorem \ref{thm:ILS.np}
Set 
\[
\BG_{\BC^{(n)}_L(\Bu),h}(E)=(\BH^{(n)}_{\BC^{(n)}_L(\Bu),h}-E)^{-1}, \quad h\in\DR,
\]
By definition a cube $\BC^{(n)}_L(\Bu)$ is $(E,m^*,0)$-NS iff 
\[
\|\Bone_{\Bu}^{(n,out)}\BG^{(n)}_{\BC^{(n)}_L(\Bu)}(E)\Bone_{\Bu}^{(n,int)}\|\leq \ee^{-\gamma(m,L_0,n)L_0}
\]
Therefore there exists sufficiently small positive $\epsilon$ such that 
\begin{equation}\label{eq:WI.epsilon}
\| \Bone_{\Bu}^{(n,out)}\BG^{(n)}_{\BC^{(n)}_L(\Bu)}(E)\Bone_{\Bu}^{(n,int)}\|\leq \ee^{-\gamma(m,L,n)L_0}-\epsilon
\end{equation}
where $m=m^*/2\in(0,\infty)$. Since by assumption $p^*\in(6Nd,\infty)$ there exists $p\in(6Nd,p^*)$ and $\tau \in(0,\infty)$ such that $L_0^{-2p4^{N-n}}-\tau\geq L_0^{-2p^*4^{N-n}}$. With such values $p$ and $\tau$ inequality \eqref{eq:ILS.NI} implies 
\begin{equation}\label{eq:WI.prob.tau}
\prob\left\{ \text{$\BC^{(n)}_L(\Bu)$ is $(E,m^*,0)$-S}\right\}\leq \frac{1}{2}L_0^{-2p4^{N-n}} -\frac{1}{2}\tau 
\end{equation}

Next, it follows from the first resolvent identity that 
\[
\| \BG^{(n)}_{\BC^{(n)}_L(\Bu),0}(E)-\BG^{(n)}_{\BC^{(n)}_L(\Bu),h}(E)\|\leq |h|\|\BU\|\cdot\|\BG^{(n)}_{\BC^{(n)}_{L_0}(\Bu)}(E)\|\cdot\|\BG^{(n)}_{\BC^{(n)}_{L_0}(\Bu)}(E)\|\cdot\|\BG^{(n)}_{\BC^{(n)}_{L_0}(\Bu),h}(E)\|.
\]
By Theorem \ref{thm:Wegner} applied to Hamiltonians $\BH^{(n)}_{\BC^{(n)}_L(\Bu),0}$ and $\BH^{(n)}_{\BC^{(n)}_L(\Bu),h}$ for any $\tau\in(0,\infty)$  there is $B(\tau)\in(0,+\infty)$ such that
\begin{align*}
&\prob\{ \|\BG^{(n)}_{\BC^{(n)}_L(\Bu),0}(E)\|\geq B(\tau)\}\leq \frac{\tau}{4}\\
&\prob\{\|\BG^{(n)}_{\BC^{(n)}_L(\Bu),h}(E)\|\geq B(\tau)\}\leq \frac{\tau}{4}.\\
\end{align*}
 Therefore 
   \begin{align*}
	 &\prob\{ \|\BG^{(n)}_{\BC^{(n)}_L(\Bu),0}(E)\|-\|\BG^{(n)}_{\BC^{(n)}_L(\Bu),h}(E)\|\geq |h| \|\BU\|B^2(\tau)\}\\
	\quad & \prob\{ \|\BG^{(n)}_{\BC^{(n)}_L(\Bu),0}(E)\|\geq B(\tau)\}+\prob\{\|\BG^{(n)}_{\BC^{(n)}_L(\Bu),h}(E)\|\geq B(\tau)\}\\
	\quad & 2\cdot \frac{\tau}{4}
	\end{align*} 
	
	Set $h^*:=\frac{\epsilon}{2\|\BU\| (B(\tau))^2}\in(0,+\infty)$. We see that if $|h|\leq h^*$, then $|h|\times \|\BU\| \times B(\tau)^2\leq \frac{\epsilon}{2}$. Hence,
	\begin{equation}\label{eq:WI.tau}
	\prob\{\|\BG^{(n)}_{\BC^{(n)}_L(\Bu),0}(E)\|-\|\BG^{(n)}_{\BC^{(n)}_L(\Bu),h}(E)\|\geq \frac{\epsilon}{2}\}\leq 2\cdot\frac{\tau}{4}
	\end{equation}
	Combining \eqref{eq:WI.epsilon}, \eqref{eq:WI.prob.tau} and \eqref{eq:WI.tau} we obtain that for all $E\in I$  
	\begin{align*}
	 &\prob\{ \text{$\BC^{(n)}_L(\Bu)$ is $(E,m,h)$-S}\}\\
	& \quad \prob\{ \text{$\BC^{(n)}_L(\Bu)$ is $(E,m,0)$-S}\}\\
	& + \prob\{ \|\BG^{(n)}_{\BC^{(n)}_L(\Bu),0}(E)\|-\|\BG^{(n)}_{\BC^{(n)}_L(\Bu),h}(E)\|\geq \frac{\epsilon}{2}\}\\
	\leq (\frac{1}{2} L_0^{-2p4^{N-n}}- \frac{\tau}{2}) +\frac{\tau}{2}=\frac{1}{2}L_0^{-2p4^{N-n}}
	\end{align*}
	\end{proof}

	\subsection{The variable energy multi-scale analysis bounds for the weakly interacting multi-particle systems}
	Here, we deduce from the fixed energy bound, the variable energy initial multi-scale analysis bound for the weakly interacting multi-particle system. We will prove localization in each compact interval $I_0$ of the following form: let $E_0\in\DR$ and $\delta=\frac{1}{2}\ee^{2L_0^{1/2}}(\ee^{-m_1L_0}-\ee^{-mL_0})$ where $m_1\in(0,m)$ by definition. Set 
	\[
	I_0:=[E-\delta,E_0+\delta].
	\]
	The result on the variable energy multi-scale analysis is given below in 
	
	\begin{theorem}\label{thm:VE.WI}
	Let $1\leq n\leq N$. For any $\Bu\in\DZ^{nd}$ we have 
	\[
	\prob\{ \text{$\exists E\in I_0$: $\BC^{(n)}_{L_0}(\Bu)$ is $(E,m_1)$-S}\}\leq L_0^{-2p4^{N-n}}
	\]
	for some $m_1\in(0,\infty)$
	\end{theorem}
	\begin{proof}
	Let $E_0\in I$. By the resolvent equation 
	\[
	\BG^{(n)}_{\BC^{(n)}_{L_0}(\Bu),h}(E)= \BG^{(n)}_{\BC^{(n)}_{L_0}(\Bu),h}(E)+(E-E_0)\BG^{(n)}_{\BC^{(n)}_L(\Bu),h}(E)\BG^{(n)}_{\BC^{(n)}_{L_0}(\Bu),h}(E_0)
	\]
	If $\dist(E,\sigma(\BH^{(n)}_{\BC^{(n)}_L(\Bu),h}(\omega)))\geq \ee^{-L_0^{2/2}}$ and $|E-E_0|\leq \frac{1}{2} \ee^{-L_0^{1/2}}$, then $\dist(E,\sigma(\BH^{(n)}_{\BC^{(n)}_{L_0}(\Bu),h}))\geq \frac{1}{2} \ee^{-L_0^{1/2}}$. 
	
	If in addition $\BC^{(n)}_{L_0}(\Bu)$ is $(E_0,m,h)$-NS then
	\[
	\|\Bone^{(n,out)}_{\Bx}\BG^{(n)}_{\BC^{(n)}_{L_0}(\Bu)}(E)\Bone_{\Bx}^{(n,int)}\|\leq \ee^{-m(1+L_0^{1/8}4^{N-n+1})L_0} + 2|E_0-E| \ee^{2L_0^{1/2}}.
	\]
	Therefore, for $m_1=\frac{m}{2} $, if we put 
	\[
	\delta=\frac{1}{2}\ee^{-L_0^{1/2}}(\ee^{-m_1(1+L_0^{1/8})^{N-n+1}L_0}-\ee^{-m(1+L_0^{1/8})^{N-n+1}})\quad I_0=[E_0-\delta, E_0+\delta],
	\] 
	we have that
	\begin{gather*}
	\prob\left\{\text{$\exists  E\in I_0$  $\BC^{(n)}_{L_0}(\Bu)$ is $(E,m_1,h)$-S}\right\}\\ \leq \prob\left\{\text{$\BC^{(n)}_{L_0}(\Bu)$ is is $(E,m,h)$-S}\right\}\\
	+\prob\left\{ \dist(E_0,\sigma(\BH^{(n)}_{\BC^{(n)}_{L_0}(\Bu),h}))\leq \ee^{-L_0^{1/2}}\right\}\\
	\leq \frac{1}{2} L_0^{-2p4^{N-n}}+ L_0^{-p4^{N}}\leq L_0^{-2p4^{N-n}}
	\end{gather*}  
	We used Theorem \ref{thm:ILS.NI} to bound the first term and the Wegner estimates Theorem  \ref{thm:Wegner} A) to bound the other term.
	\end{proof}   
	Below, we develop the induction step of the multi-scale  and for the reader convenience, we also give the proof of some important results.

\section{Multi-scale induction}\label{sec:MP.induction}
In the rest of the paper, we assume that $n\geq 2$ and $I_0$  is the interval of the previous Section. Recall the following facts from \cite{Eka19a}. Consider a cube  $\BC^{(n)}_L(\Bu)$ with $\Bu=(u_1,\ldots,u_n)\in(\DZ^d)^n$. We have 
\[
\varPi\Bu=\{u_1,\ldots,n\}
\]
and 
\[
\varPi\BC^{(n)}_{L_0}(\Bu)=C^{(1)}_{L_0}(u_1)\cup \cdots\cup C^{(1)}_{L_0}(u_n)
\]
\begin{definition}
Let $L_0\geq 3$ be a constant and $\alpha=3/2$. We define the sequence $\{L_k, k\geq 1\}$  recursively as follows 
\[
L_k=\lfloor L_{k-1}\rfloor+1,\quad \text{for all $k\geq 1$}.
\]
\end{definition}
Let $m\in(0,\infty)$ be a positive constant, we also introduce the following property, namely the multi-scale analysis bounds at any scale length $L_k$ and for any pair of separable cubes $\BC^{(n)}_{L_k}(\Bu)$ $\BC^{(n)}_{L_k}(\Bv)$ 
\begin{dsknN*}
\[
\prob\left\{\text{$\exists E\in I_0$ $\BC^{(n)}_{L_k}(\Bu)$ $\BC^{(n)}_L(\Bv)$  are $(E,m)$-S}\right\}\leq L_k^{-2p4^{N-n}}
\]
where $p\in(6Nd,\infty)$.
\end{dsknN*}
In both the single-particle and the multi-particle system, given the results of the multi-particle multi-scale analysis property $\dsknN$ above, one can deduce the localization results see for example the papers \cites{DK89,DS01} for those concerning the single-particle case and \cites{CS09,Eka19a}  for multi-particle systems. We have the following:

\begin{theorem}\label{thm:ILS}
For any $n'\in(1,n)$ assume that property $\dskprimenN$ holds true for all $k\geq 0$ then there exists a positive constant $\tilde{\mu}\in(0,\infty)$  such that for cube $\BC^{(n')}_{L_k}(\Bu')$
\begin{equation}\label{eq:ILS}
\esm\left[\|\Bone_{\BC^{(n',out)}_L(\Bu')}\BG^{(n')}_{\BC^{(n')}_L(\Bu')}(E)\Bone^{(n',int)}_{\BC^{(n',int)}(\Bu')}\|\right]\leq \ee^{-\tilde{\mu}L}
\end{equation}
\end{theorem}

\begin{definition}[partially/fully interactive]\label{def:PI.FI}
An $n$-particle cube $\BC^{(n)}_L(\Bu)\subset\DZ^{nd}$ is called fully interactive (FI) if 
\begin{equation}\label{eq:PI.FI}
 \diam \varPi \Bu:=\max_{i\neq j} |u_i-u_j|\leq n(2L+r_0),
\end{equation}
and partially interactive (PI) otherwise.
\end{definition}

The following simple statement clarifies the notion of PI cubes 
\begin{lemma}\label{lem:PI}
If a cube $\BC^{(n)}_L(\Bu)$ is PI then there exists a subset $\CJ\subset\{1,\ldots,n\}$ with $1 \leq \card \CJ\leq n-1$ such that
\[
\dist\left( \varPi_{\CJ}\BC^{(n)}_L(\Bu),\varPi_{\CJ^{c}}\BC^{(n)}_L(\Bu) \right)\geq r_0
\]
\end{lemma}
\begin{proof}
See the proof in the  appendix Section \ref{sec:appendix}
 \end{proof}
If a cube $\BC^{(n)}_L(\Bu)$ is PI then by Lemma \ref{lem:PI}, we can write it as 
 \begin{equation}\label{eq:PI.cubes}
\BC^{(n)}_L(\Bu)=\BC^{(n')}_L(\Bu')\times \BC^{(n'')}_L(_Bu'')
\end{equation}
with 
\begin{equation}\label{eq:distant.PI}
\dist(\varPi\BC^{(n')}_L(\Bu'),\varPi\BC^{(n'')}_L(\Bu''))\geq r_0
\end{equation}
where $\Bu'=\Bu_{\CJ}=(u_j:j\in\CJ)$ $\Bu''=\Bu_{\CJ^{c}}=(u_j:j\in\CJ^{c})$ $n'=\card \CJ$ and $n''=\card \CJ^{c}$
Throughout, when we write a PI cube $\BC^{(n)}_L(\Bu)$ in the form \eqref{eq:PI.cubes}, we implicitly assume that the projections satisfy \eqref{eq:distant.PI}. Let $\BC^{(n')}_L(\Bu')\times\BC^{(n'')}_L(\Bu'')$ be the decomposition of the PI cube $\BC^{(n)}_L(\Bu)$ and $\{\lambda_i,\varphi_i\}$ and $\{\mu_j,\phi_j\}$ be the eigenvalues and corresponding eigenfunctions of $\BH^{(n')}_{\BC^{(n')}_L(\Bu')}$ and $\BH^{(n'')}_{\BC^{(n'')}_L(\Bu'')}$ respectively. Next, we can choose the eigenfunctions $\BPsi_{ij}$ as tensor product
\[
\BPsi_{ij}=\varphi_i\otimes \phi_j
 \]
The eigenfunctions appearing in subsequent argument and calculations will be assume normalized.

Now we turn to geometrical property of FI cubes 

\begin{lemma}\label{lem:FI.cubes }
Let $n\geq 1$, $L\geq 2r_0$ and consider two FI cubes $\BC^{(n)}_L(\Bx)$ and $\BC^{(n)}_L(\By)$ with $|\Bx-\By|\geq 7nL$. Then 
\[
\varPi\BC^{(n)}_L(\Bx)\cap \varPi\BC^{(n)}_L(\By)=\emptyset
\]
\end{lemma}

\begin{proof}
See the proof in the Appendix Section \ref{sec:appendix}.
\end{proof}

Given an $n$-particle cube $\BC^{(n)}_L(\Bu)$ and $E\in \DR$, we denote by

\begin{itemize}
\item 
$M^{\sep}_{\PI}(\BC^{(n)}_{L_{k+1}}(\Bu),E)$ the maximal number of pairwise separable $(E,m)$-singular PI cubes $\BC^{(n)}_{L_k}(\Bu^{(j)})\subset \BC^{(n)}_{L_{k+1}}(\Bu)$;\\
\item
by $M_{\PI}(\BC^{(n)}_{L_{k+1}}(\Bu),E)$, the maximal number of (not necessary separable) $(E,m)$-singular PI-cubes $\BC^{(n)}_{L_k}(\Bu^{(j)})$ contain in $\BC^{(n)}_{L_{k+1}}(\Bu)$ with $\Bu^{(j)}, \Bu^{(j')}\DZ^{nd}$ and $|\Bu^{(j)}-\Bu^{(j')}|\geq 7NL_k$ for all $j\neq j'$;\\
\item
$M_{\FI}(\BC^{(n)}_{L_{k+1}}(\Bu),E)$ the maximal number of $(E,m)$-singular FI cubes  $\BC^{(n)}_{L_k}(\Bu^{(j)})\subset \BC^{(n)}_{L_{k+1}}(\Bu)$ with $|\Bu^{(j)}-\Bu^{(j')}|\geq 7NL_k$ for all $j\neq j'$\footnote{Note that by Lemma \ref{lem:FI.cubes}; two FI cubes $\BC^{(n)}_{L_k}(\Bu^{(j)}) $ and $\BC^{(n)}_{L_k}(\Bu^{(j')})$ with $|\Bu^{(j)}-\Bu^{(j')}|\geq 7NL_k$ are automatically separable.};\\
\item
$M_{\PI}(\BC^{(n)}_{L_{k+1}}(\Bu),I):= \sup_{E\in I} M_{\PI}(\BC^{(n)}_{L_{k+1}}(\Bu),E)$;\\
\item
$M_{\FI}(\BC^{(n)}_{L_{k+1}}(\Bu),I):= \sup_{E\in I} M_{\FI}(\BC^{(n)}_{L_{k+1}}(\Bu),E)$;\\
\item
$M(\BC^{(n)}_{L_{k+1}}(\Bu),E)$ the maximal number of $(E,m)$-singular cubes $\BC^{(n)}_{L_k}(\Bu^{(j)})\subset\BC^{(n)}_{L_{k+1}}(\Bu)$ with $\dist(\Bu^{(j)},\partial\BC^{(n)}_{L_{k+1}}(\Bu))\geq 2L_k$ and $|\Bu^{(j)}-\Bu^{(j')}|\geq 7NL_k$ for all $j\neq j'$;\\
\item
$M^{\sep}(\BC^{(n)}_{L_{k+1}}(\Bu),E)$ the maximal number of pairwise separable $(E,m)$-singular cube $\BC^{(n)}_{L_k}(\Bu^{(j)})\subset \BC^{(n)}_{L_{k+1}}(\Bu)$;
\end{itemize}

Clearly, 
\[
M_{\PI}(\BC^{(n)}_{L_{k+1}}(\Bu),E)+M_{\FI}(\BC^{(n)}_{L_{k+1}}(\Bu),E)\geq M(\BC^{(n)}_{L_{k+1}}(\Bu),E).
\]

\subsection{Pairs of partially interactive cubes} \label{sec:PI.cubes}

Let $\BC^{(n)}_{L_{k+1}}(\Bu)=\BC^{(n')}_{L_{k+1}}(\Bu')\times\BC^{(n'')}_{L_{k+1}}(\Bu'')$ be a PI-cube. We also write $\Bx=(\Bx',\Bx'')$ for any point $\Bx\in\BC^{(n)}_{L_{k+1}}(\Bu)$, in the same way as $(\Bu', \Bu'')$. So the corresponding  Hamiltonian $\BH^{(n)}_{\BC^{(n)}_{L_{k+1}}(\Bu)}$ is written in the form:

\[
\BH^{(n)}_{\BC^{(n)}_{L_{k+1}}(\Bu)}\BPsi(\Bx)=(-\BDelta\BPsi)(\Bx) +\left[\BU(\Bx')+\BV(\Bx',\omega) + \BU(\Bx'') +\BV(\Bx'',\omega)\right]\BPsi(\Bx)
\]
or in compact form:  
\[
\BH^{(n)}_{\BC^{(n)}_{L_{k+1}}(\Bu)}=\BH^{(n')}_{\BC^{(n')}_{L_{k+1}}(\Bu')}\otimes\BI+\BI\otimes\BH^{(n'')}_{\BC^{(n'')}_{L_{k+1}}(\Bu'')}
\]

\begin{definition}\label{def:localized}
Let $n\geq 2$ and $\BC^{(n')}_{L_k}(\Bu')\times\BC^{(n'')}_{L_k}(\Bu'')$ be the decomposition of the PI cube $\BC^{(n)}_{L_k}(\Bu)$. Then $\BC^{(n)}_{L_k}(\Bu)$ is called
\begin{enumerate}
 \item[(i)]
$m$-left-localized if for any  normalized eigenfunction $\varphi^{(n')}$ of the restricted Hamiltonian $\BH^{(n')}_{\BC^{(n')}_L( \Bu')}(\omega)$, we have
\[
\| \Bone_{\BC^{(n',out)}(\Bu')}\varphi^{(n')}\|\leq \ee^{-2\gamma(m,L_k,n')L_k}
\]
otherwise it is called $m$-non-left-localized,\\
\item[(ii)]
$m$-right-localized if for any normalized eigenfunction $\varphi^{(n'')}$ of the restricted Hamiltonian $\BH^{(n'')}_{\BC^{(n'')}_L( \Bu'')}(\omega)$, we have
\[
\| \Bone_{\BC^{(n'',out)}(\Bu'')}\varphi^{(n'')}\|\leq \ee^{-2\gamma(m,L_k,n'')L_k}
\]
otherwise it is called $m$-non-right localized,\\
\item[(iii)]
$m$-localized if it is $m$-left-localized and $m$-right-localized. Otherwise it is called $m$-non-localized
\end{enumerate}
\end{definition}
\begin{lemma}\label{lem:NR.NS}
Let $E\in I$ and $\BC^{(n)}_{L_k}(\Bu)$ be a PI cube. Assume that $\BC^{(n)}_{L_k}(\Bu)$ is $E$-NR and $m$-localized. Then the cube $\BC^{(n)}_{L_k}(\Bu)$ is $(E,m)$-NS.
\end{lemma}

\begin{proof}
We proceed as in Lemma \ref{lem:NR.NS.gamma}.
\end{proof}
Now, before proving the main results of this Subsection concerning the probability of two PI cubes to be singular at the same energy we need first to estimate the one for a non-localized cube given in the statement below
\begin{lemma}\label{lem:loc.prob}
Let $\BC^{(n)}_{L_k}(\Bu)$ be a PI cube. Then 
\[
\prob\{\text{$\BC^{(n)}_{L_k}(\Bu)$ is $m$-non-localized}\}\leq L_k^{-4p4^{N-n}}.
\]
\end{lemma}

\begin{proof}
The proof combines the ideas of Theorem \ref{thm:ILS.np} in the multi-particle systems without interaction and the induction assertion of localization given in Theorem \ref{thm:ILS}.
\end{proof}
Now, we state the main result of this Subsection, i.e., the probability bound of two PI cubes to be singular at the same energy belonging to the compact  interval $I_0$ introduced at the begening of the Section.

\begin{theorem}\label{thm:PI.cubes}
Let $2\leq n\leq N$. There exists $L_1^*=L_1^*(N,d)\in(0,\infty)$ such that if $L_0\geq L_1^*$ and if for $k\geq 0$ $\dsknprimeN$ holds true for any $n'\in(1,n)$ then $\dskonN$ holds for any pair of separable PI cubes $\BC^{(n)}_{L_{k+1}}(\Bx)$ and $\BC^{(n)}_{L_{k+1}}(\By)$.
\end{theorem} 
\begin{proof}
Let $\BC^{(n)}_{L_{k+1}}(\Bx)$ and $\BC^{(n)}_{L_{k+1}}(\By)$ be two separable PI cubes. Consider the events:
\begin{align*}
 \rm{B}_{k+1}&=\{\text{$\exists E\in I_0:$ $\BC^{(n)}_{L_{k+1}}(\Bx)$ $\BC^{(n)}_{L_{k+1}}(\By)$ are $(E,m)$-S}\},\\
 \rm{R}&=\{\text{$\exists E\in I_0$: $\BC^{(n)}_{L_{k+1}}(\Bx)$  and $\BC^{(n)}_{L_{k+1}}(\By)$ are $E$-R}\},\\
\CN_{\Bx}&=\{\text{$\BC^{(n)}_{L_{k+1}}(\Bx)$ is $m$-non-localized}\},\\
\CN_{\By}&=\{\text{$\BC^{(n)}_{L_{k+1}}\By)$ is $m$-non-localized}\}\\
\end{align*} 
If $\omega\in\rm{B}_{k+1}\setminus \rm{R}$ then $\forall \in I_0$, $\BC^{(n)}_{L_{k+1}}(\Bx)$ or $\BC^{(n)}_{L_{k+1}}(\By)$ is $E$-NR, then it must be $m$-non-localized: otherwise it would have been $(E,m)$-NS by Lemma \ref{lem:NR.NS}. Similarly if $\BC^{(n)}_{L_{k+1}}$ is $E$-NR,  then it must be $m$-non-localized. This implies that 
\[
\rm{B}_{k+1}\subset \rm{R}\cup \CN_{\Bx}\cup\CN_{\By}
\] 
Therefore, using Theorem \ref{thm:Wegner} and Lemma \ref{lem:loc.prob}, we have  
\begin{align*}
\prob\{\rm{B}_{k+1}\}&\leq \prob\{\rm{R}+\prob\{\CN_{\Bx}\}+ \prob\{\CN_{\By}\}\\
&L_{k+1}^{-p4^{N-n}}+\frac{1}{2}L_{k+1}^{-4p4^{N-n}}+\frac{1}{2}L_{k+1}^{-4p4^{N-n}}
\end{align*}
Finally
\begin{equation}\label{eq:PI.prob}
\prob\{\rm{B}_{k+1}\}\leq L_{k+1}^{-p4^{N-n}}+L_{k+1}^{-4p4^{N-n}}\leq L_{k+1}^{-2p4^{N-n}}
\end{equation}
which proves the result.
\end{proof}

For subsequent calculations and proofs we give the following two Lemmas:

\begin{lemma}\label{lem:M}
If $M(\BC^{(n)}_{L_{k+1}}(\Bu),E)\geq\kappa(n)+2$ with $\kappa(n)=n^n,$ then $M^{\sep}(\BC^{(n)}_{L_{k+1}}(\Bu),E)\geq 2$. Similarly if $M_\PI(\BC^{(n)}_{L_{k+1}}(\Bu),E)\geq\kappa(n) +2$ then $M^{\sep}_\PI(\BC^{(n)}_{L_{k+1}}(\Bu),E)\geq 2$.
\end{lemma}

\begin{proof}
See the appendix Section \ref{sec:appendix}.
\end{proof}
\begin{lemma}\label{lem:prob.M}
With the above notations, assume that $\dskunnprimeN$ holds true for all $n'\in[1,n)$ then 
\[
\prob\left\{M_\PI(\BC^{(n)}_{L_{k+1}}(\Bu),I)\geq \kappa(n)+2\right\}\leq \frac{3^{2nd}}{2}L_{k+1}^{2nd}\left(L_k^{-4^Np}+L_k^{-4p4^{N-n}}\right)
\]
\end{lemma}
\begin{proof}
See the appendix Section \ref{sec:appendix}
\end{proof}

\subsection{Pairs of fully interactive cubes} \label{sec:FI.cubes}
Our aim now is to prove $\dskonN$ for a pair of fully interactive cubes $\BC^{(n)}_{L_{k+1}}(\Bx)$n  and $\BC^{(n)}_{L_{k+1}}(\By)$. We adapt to the continuum a very crucial and hard result obtained in the paper \cite{Eka19a} and which generalized to multi-particle systems  some previous work by von Dreifus and Klein \cite{DK89} on the lattice and Stollmann \cite{Sto01} in the continuum for single particle models.

\begin{lemma}\label{lem:CNR.NS}
Let $J=\kappa(n)+5$ with $\kappa(n)=n^n$ and $E\in\DR$. Suppose that 
\begin{enumerate}
\item[i)] $\BC^{(n)}_{L_{k+1}}(\Bx)$ is $E$-CNR.\\
\item[ii)] $M(\BC^{(n)}_{L_{k+1}}(\Bx),E)\leq J$.
\end{enumerate}
Then there exists $\tilde{L}_2^*(J,N,d)\geq 0$ such that if $L_0\geq \tilde{L}^*_2(J,N,d)$ we have that $\BC^{(n)}_{L_{k+1}}(\Bx)$ is $(E,m)$-NS
\end{lemma}

\begin{proof}
Since $M(\BC^{(n)}_{L_{k+1}}(\BX),E)\leq J$, there exist at most $J$ cubes of side length $2L_k$ contained in $\BC^{(n)}_{L_{k+1}}(\Bx)$ that are $(E,m)$-S with centers at distance $\geq 7NL_k$. Therefore, we can find $\Bx_i\in\BC^{(n)}_{L_{k+1}}(\Bx)\cap \Gamma_{\Bx}$ with $\Gamma_{\Bx}=\Bx+\frac{L_k}{3}\DZ^{nd}$.
\[
  \dist(\Bx_i,\partial\BC^{(n)}_{L_{k+1}}(\Bx))\geq 2L_k,\quad i=1,\ldots,r\leq J
	\]
	
	such that, if $\Bx_0\in\BC^{(n)}_{L_{k+1}}(\Bx)\setminus \bigcup_{i=1}^r\BC^{(n)}_{2L_k}(\Bx_i)$, then the cube $\BC^{(n)}_{L_k}(\Bx_0)$ is $(E,m)$-NS. 
	
We do an induction procedure in $\BC^{(n,int)}_{L_{k+1}}(\Bx)$ and start with $\Bx_0\in\BC^{(n,int)}_{L_{k+1}}(\Bx)$. We estimate $\|\Bone_{\BC^{(n,out)}_{L_{k+1}}(\Bx)}\BG^{(n)}_{L_{k+1}}(E)\Bone_{\BC^{(n,int)}_{L_k}(\Bx_0)}\|$. Suppose that $\Bx_0, \ldots\Bx_{\ell}$ have been choosen for $\ell\geq 0$ We have two cases

\begin{enumerate}
\item[case a)] $ \BC^{(n)}_{L_k}(\Bx_{\ell})$ is $(E,m)$-NS

In this case, we apply the (GRI) Theorem \ref{thm:GRI} and obtain 
\begin{gather*}
\|\Bone_{\BC^{(n,out)}_{L_{k+1}}(\Bx)}\BG^{(n)}_{\BC^{(n)}_{L_{k+1}}(\Bx)}(E)\Bone_{\BC^{(n,int)}_{L_{k+1}}(\Bx_0)}\|\\
\leq Cgeom \|\Bone_{\BC^{(n,out)}_{L_{k+1}}(\Bx)}\BG^{(n)}_{\BC^{(n)}_{L_{k+1}}(\Bx)}(E)\Bone_{\BC^{(n,out)}_{L_{k+1}}(\Bx_0)}\|\cdot \\
\|\Bone_{\BC^{(n,out)}_{L_{k+1}}(\Bx)}\BG^{(n)}_{\BC^{(n)}_{L_{k+1}}(\Bx_0)}(E)\Bone_{\BC^{(n,int)}_{L_{k+1}}(\Bx_0)}\|\\
\leq Cgeom \|\Bone_{\BC^{(n,out)}_{L_{k+1}}(\Bx)}\BG^{(n)}_{\BC^{(n)}_{L_{k+1}}(\Bx)}(E)\Bone_{\BC^{(n,out)}_{L_{k+1}}(\Bx)}\|\cdot \ee^{-\gamma(m,L_k,n)L_k}.
\end{gather*}
We replace in the above analysis $\Bx$ with $\Bx_{\ell}$ and we get
\begin{gather*}
\|\Bone_{\BC^{(n,out)}_{L_{k+1}}(\Bx_{\ell})}\BG^{(n)}_{\BC^{(n)}_{L_{k+1}}(\Bx_{\ell})}(E)\Bone_{\BC^{(n,int)}_{L_{k+1}}(\Bx_{\ell})}\|\\
\leq 3^{nd}\|\Bone_{\BC^{(n,out)}_{L_{k+1}}(\Bx_{\ell})}\BG^{(n)}_{\BC^{(n)}_{L_{k+1}}(\Bx_{\ell})}(E)\Bone_{\BC^{(n,int)}_{L_{k+1}}(\Bx_{\ell+1})}\|,
\end{gather*}
where $\Bx_{\ell+1}$ is choosen in such a way that the norm in the right hand side in the above equation is maximal. Observe that $|\Bx_{\ell}-\Bx_{\ell+1}|=L_k/3$. We therefore obtain 
\begin{gather*}
\|\Bone_{\BC^{(n,out)}_{L_{k+1}}(\Bx)}\BG^{(n)}_{\BC^{(n)}_{L_{k+1}}(\Bx_{\ell})}(E)\Bone_{\BC^{(n,int)}_{L_{k+1}}(\Bx_{\ell})}\|\\
\leq C_{geom} 3^{nd}\ee^{-\gamma(m,L_k,n)L_k}\|\Bone_{\BC^{(n,out)}_{L_{k+1}}(\Bx)}\BG^{(n)}_{\BC^{(n)}_{L_{k+1}}(\Bx)}(E)\Bone_{\BC^{(n,int)}_{L_{k+1}}(\Bx_{\ell+1})}\| \\
\leq \delta_+\|\Bone_{\BC^{(n,out)}_{L_{k+1}}(\Bx)}\BG^{(n)}_{\BC^{(n)}_{L_{k+1}}(\Bx)}(E)\Bone_{\BC^{(n,int)}_{L_{k+1}}(\Bx_{\ell+1})}\\
\end{gather*}
with $\delta_+ =3^{nd} C_{geom}\ee^{-\gamma(m,L_k,n)L_k}$.\\  
\item[case (b)]
$\BC^{(n)}_{L_k}(\Bx_{\ell})$ is $(E,m)$-S.
Thus, there exists $i_0=1,\ldots,r$ such that $\BC^{(n)}_{L_k}(\Bx_{\ell})\subset \BC^{(n)}_{2L_k}(\Bx_{i_0})$. We apply again the (GRI) this time with $\BC^{(n)}_{L_{k+1}}(\Bx)$ and $\BC^{(n)}_{2L_k}(\Bx_{i_0})$ and obtain
\begin{gather*}
\|\Bone_{\BC^{(n,out)}_{L_{k+1}}(\Bx)}\BG^{(n)}_{\BC^{(n)}_{L_{k+1}}(\Bx)}(E)\Bone_{\BC^{(n,int)}_{2L_k}(\Bx_{i_0})}\|\leq C_{geom} \|\Bone_{\BC^{(n,out)}_{L_{k+1}}(\Bx)}\BG^{(n)}_{\BC^{(n)}_{L_{k+1}}(\Bx)}(E)\Bone_{\BC^{(n,out)}_{L_{k+1}}(\Bx_{i_0})}\| \\
 \times \|\Bone_{\BC^{(n,out)}_{L_{k}}(\Bx_{i_0})}\BG^{(n)}_{\BC^{(n)}_{L_{k}}(\Bx_{i_0})}(E)\Bone_{\BC^{(n,int)}_{L_{k}}(\Bx_{i_0})}\|\\
\leq C_{geom}\ee^{(2L_k)^{1/2}}\cdot \|\Bone_{\BC^{(n,out)}_{L_{k+1}}(\Bx)}\BG^{(n)}_{\BC^{(n)}_{L_{k+1}}(\Bx)}(E)\Bone_{\BC^{(n,out)}_{2L_k}(\Bx_{i_0})}\|\\
\end{gather*}
We have almost everywhere 
\[
\Bone_{\BC_{2L_k}^{(n,out)}(\Bx_{i_0})}\sum_{\tilde{\Bx}\in\BC^{(n)}_{2L_k}(\Bx_{i_0})\cap\Gamma_{\Bx_{i_0}}, \BC^{(n)}_{L_k}(\tilde{\Bx})\not\subset \BC^{(n)}_{2L_k}(\Bx_{i_0})} \Bone_{\BC^{(n,int)}_{L_k}(\Bx)}
\]
Hence, by choosing $\tilde{\Bx}$ is such a way that the right hand side is maximal, we get 
\[
\|\Bone_{\BC^{(n,out)}_{L_{k+1}}(\Bx)}\BG^{(n)}_{\BC^{(n)}_{L_{k+1}}(\Bx)}(E)\Bone_{\BC^{(n,int)}_{2L_k}(\Bx_{i_0})}\|\leq 6^{nd}\cdot\|\Bone_{\BC^{(n,out)}_{L_{k+1}}(\Bx)}\BG^{(n)}_{\BC^{(n)}_{L_{k+1}}(\Bx)}(E)\Bone_{\BC^{(n,int)}_{L_{k+1}}(\tilde{\Bx})}\|.
\]
Since $\BC^{(n)}_{L_k}(\tilde{\Bx})\not\subset\BC^{(n)}_{2L_k}(\Bx_{i_0})$, $\tilde{\Bx}\in\BC^{(n)}_{2L_k}(\Bx_{i_0})$  and the cubes $\BC^{(n)}_{2L_k}(\Bx_i)$ are disjoint, we obtain that 
\[
\BC^{(n)}_{L_k}(\tilde{\Bx})\not\subset \bigcup_{i=1}^r \BC^{(n)}_{2L_k}(\Bx_i).
\]
so that the cube $\BC^{(n)}_{L_k}(\tilde{\Bx})$ must be $(E,m)$-NS. We therefore perform a new step as in case (a) and obtain
\[
\ldots \leq 6^{nd}3^{nd}C_{geom}\ee^{-\gamma(m,L_k,n)L_k}\cdot \|\Bone_{\BC^{(n,out)}_{L_{k+1}}(\Bx)}\BG^{(n)}_{\BC^{(n)}_{L_{k+1}}(\Bx)}(E)\Bone_{\BC^{(n,int)}_{L_{k+1}}(\Bx_{\ell+1}}\|,
\]
with $\Bx_{\ell+1}\in\Gamma_{\Bx}$ and $|\tilde{\Bx}-\Bx_{\ell+1}| =L_k/3$.

Summarizing, we get $\Bx_{\ell+1}$ with 
\[
\|\Bone_{\BC^{(n,out)}_{L_{k+1}}(\Bx)}\BG^{(n)}_{\BC^{(n)}_{L_{k+1}}(\Bx)}(E)\Bone_{\BC^{(n,int)}_{L_k}(\Bx_{\ell})}\|\leq \delta_0 \|\Bone_{\BC^{(n,out)}_{L_{k+1}}(\Bx)}\BG^{(n)}_{\BC^{(n)}_{L_{k+1}}(\Bx)}(E)\Bone_{\BC^{(n,int)}_{L_{k+1}}(\Bx_{\ell+1})}\|,
\]
with $\delta_0=18^{nd}C_{geom}^2\ee^{(2L_k)^{1/2}}\ee^{-\gamma(m,L_k,n)L_k}$ After $\ell$ iterations with $n_+$ steps of case (a) and $n_0$ steps of case (b), we obtain
\begin{gather*}
\|\Bone_{\BC^{(n,out)}_{L_{k+1}}(\Bx)}\BG^{(n)}_{\BC^{(n)}_{L_{k+1}}(\Bx)}(E)\Bone_{\BC^{(n,int)}_{L_{k+1}}(\Bx_{0})}\|\leq (\delta_+)^{n_+}(\delta_0)^{n_0}\\ 
\times\|\Bone_{\BC^{(n,out)}_{L_{k+1}}(\Bx)}\BG^{(n)}_{\BC^{(n)}_{L_{k+1}}(\Bx)}(E)\Bone_{\BC^{(n,int)}_{L_{k}}(\Bx_{\ell})}\|.
\end{gather*}
Now since $\gamma(m,L_k,n)\geq m$ we have that 
\[
\delta_+\leq 3^{nd}\cdot C_{geom}\ee^{-mL_k}.
\]
So $\delta_+$ can be made arbitrarily small if $L_0$ and hence $L_k$ is large enough. We also have for $\delta_0$

\begin{align*}
\delta_0&=18^{nd} C^2_{geom}\ee^{(2L_k)^{1/2}}\ee^{-\gamma(m,L_k,n)L_k}\\
&18^{nd}C^2_{geom}\ee^{\sqrt{2}L_k^{1/2}}\ee^{-\gamma(m,L_k,n)L_k}\\
&\leq 18^{nd} C^2_{geom}\ee^{\sqrt{2}L_k^{1/2}-mL_k}\leq \frac{1}{2}.
\end{align*}

For large $L_0$ hence $L_k$. Using the (GRI), we can iterate if $\BC^{(n,out)}_{L_{k+1}}(\Bx)\cap\BC^{(n)}_{L_k}(\Bx_{\ell})=\emptyset$. Thus, we can have at least $n_+$ steps of case (a)  with 
\[
n_+\cdot\frac{L_k}{3}+ \sum_{i=1}^r 2L_k\geq \frac{L_{k+1}}{3}-\frac{L_k}{3},
\]
until the induction eventually stop. Since $r\leq J$, we can bound $n_+$ from below:
\begin{align*}
n_+\cdot\frac{L_k}{3}&\geq \frac{L_{k+1}}{3}-\frac{L_k}{3}-r(L_k)\\ 
&\geq \frac{L_{k+1}}{3}-\frac{L_k}{3}-2JL_k
\end{align*}

Which yields 
\begin{align*}
n_+&\geq \frac{L_{k+1}}{L_k}-1-6J\\
&\geq \frac{L_{k+1}}{L_k}-7J
\end{align*}
Therefore 
\begin{equation}\label{eq:delta.plus}
\|\Bone_{\BC^{(n,out)}_{L_{k+1}}(\Bx)}\BG^{(n)}_{\BC^{(n)}_{L_{k+1}}(\Bx)}(E)\Bone_{\BC^{(n,int)}_{L_{k}}(\Bx_{0})}\|\leq \delta_+^{n_+}\cdot\|\BG^{(n)}_{\BC^{(n)}_{L_{k+1}}(\Bx)}(E)\|
\end{equation}
Finally, by $E$-non-resonance of $\BC^{(n)}_{L_{k+1}}(\Bx)$ and since we can cover $\BC^{(n,int)}_{L_{k+1}}(\Bx)$ by $\left(\frac{L_{k+1}}{L_k}\right)^{nd}$ small cubes $\BC^{(n,int)}_{L_k}(\By)$, equation \eqref{eq:delta.plus} with $y$ instead of $\Bx_0$, yields 
\begin{gather*}
\|\Bone_{\BC^{(n,out)}_{L_{k+1}}(\Bx)}\BG^{(n)}_{\BC^{(n)}_{L_{k+1}}(\Bx)}(E)\Bone_{\BC^{(n,int)}_{L_{k+1}}(\Bx)}\|\\
\leq \left(\frac{L_{k+1}}{L_k}\right)\cdot \delta_{+}^{n_+}\cdot\ee^{L_{k+1}^{1/2}}\\
\leq \left(\frac{L_{k+1}}{L_k}\right)\cdot\left[3^{nd}\cdot C_{geom}\cdot \ee^{-\gamma(m,L_k,n)}\right]^{\frac{L_{k+1}}{L_k}-7J}\times \ee^{L_{k+1}^{1/2}}\\
\leq L_{k+1}^{nd}L_{k+1}^{-\frac{nd}{\alpha}}C(n,d)^{\frac{L_{k+1}}{L_k}-7J}\ee^{-\gamma(m,L_k,n)(\frac{L_{k+1}}{L_k}-7J)}\times \ee^{L_{k+1}^{1/2}}\\
\leq L_{k+1}^{nd/3}\ee^{(L_{k+1}^{1/3}-7J)\ln C(n ,d)}\ee^{-\gamma(m,L_k,n)(L_{k+1}^{1/3}-7J)}\ee^{L_{k+1}^{1/2}}\\
\leq\ee^{-\left[\frac{-nd}{3}\ln(L_{k+1})-L_{k+1}^{1/3}\ln(C)+7J\ln(C(n,d))+\gamma(m,L_k,n)L_{k+1}^{1/3}-7J\gamma(m,L_k,n)-L_{k+1}^{1/2}\right]}\\
\leq \ee^{-\left[\frac{-nd}{3}\frac{\ln L_{k+1}}{L_{k+1}}-\frac{L_{k+1}^{1/3}\ln(C(n,d))}{L_{k+1}}+\frac{7J\ln(C(n,d))}{L_{k+1}}+\gamma(m,n,L_k)\frac{L_{k+1}^{1/3}}{L_{k+1}}-7J\frac{\gamma(m,L_k,n)}{L_{k+1}}-L_{k+1}^{-1/2}\right]}\\
\leq \ee^{-m'L_{k+1}}           
\end{gather*}  
where 
\[
m'=\frac{1}{L_{k+1}}\left[n_+\gamma(m,L_k,n)L_k-n_+\ln(2^{Nd} NdL_k^{nd-1})\right]- \frac{1}{L_{k+1}^{1/2}},
\]
with
\[
L_{k+1}L_k^{-1}-7J\leq n_+\leq L_{k+1} L_k^{-1}
\]
we obtain
\begin{align*}
m'&\geq\gamma(m,n,L_k)-\gamma(m,L_k,n)\frac{7JL_k}{L_{k+1}}\\
&-\frac{1}{L_{k+1}}\frac{L_{k+1}}{L_k}\ln((2_{Nd}Nd)L_k^{nd-1})-\frac{1}{L_{k+1}^{1/2 }}\\
&\geq\gamma(m,L_k,n)-\gamma(m,L_k,n)7JL_k^{-1/2}\\
&\quad -L_k^{-1}(\ln(2^{Nd}Nd))-(nd-1)\ln(L_k)-L_k^{-3/4}\\
&\geq\gamma(m,L_k,n)\left[1-(7J+\ln(2^{Nd}Nd)+Nd)L_k^{-1/2}\right]\\
\end{align*} 
if $L_0\geq L_2^*(J,N,d)$ for some $L_2^*(J,N,d)\geq 0$ large enough. Since $\gamma(m,L_k,n)=m(1+L_k^{-1/8})^{N-n+1}$
\[
\frac{\gamma(m,L_k,n)}{\gamma(m,L_{k+1},n)}=\left(\frac{1+L_k^{-1/8}}{1+L_k^{-3/16}}\right)^{N-n+1}\geq \frac{1+L_k^{-1/8}}{1+L_k^{-3/16}}.
\]
Therefore we can compute
\begin{align*}
&\frac{\gamma(m,L_k,n)}{\gamma(m,L_{k+1},n)}\left(1-(7J+\ln(2^{Nd}Nd)+Nd)L_k^{-1/2}\right)\\
&\quad\frac{1+L_k^{-1/8}}{1+L_k^{-3/16}}\left(1-(7J+\ln(2^{Nd}Nd)+Nd)L_k^{-1/2}\right)\geq 1
\end{align*}
provided$L_0\geq\tilde{L}_2^*$ for some large enough $\tilde{L}_2^*(J,N,d)\geq 0$. Finally,we obtain that $m'\geq \gamma(m,L_{k+1},n)$.  This proves the result.
\end{enumerate}
\end{proof}  
\begin{lemma}\label{lem:prob.M2}
Given $k\geq 0$, asssume that property $\dsknN$ holds true for all pairs of separable FI cubes. Then for any $\ell\geq 1$
\[
\prob\left\{M_\FI(\BC^{(n)}_{L_{k+1}}(\Bu),I)\geq 2\ell\ \right\}\leq C(n,N,d,\ell)L_k^{2\ell dn \alpha}L_k^{-2\ell p4^{N-n}}
\]
\end{lemma}
\begin{proof}
See the proof in the appendix Section \ref{sec:appendix}. 
\end{proof}

\begin{theorem} \label{thm:FI.cubes}
Let $1\leq n\leq N$. There exists $L_2^*=L_2^*(N,d)\geq 0$ such that  if $L_0\geq L_2^*$ and if for $k\geq 0$

\begin{enumerate}
\item[(i)] $\dskunnprimeN$ for all $n'\in[1,n)$ holds true,
\item[(ii)] $\dsknN$ holds true for all pairs of FI cubes
\end{enumerate}
then $\dskonN$ holds true for any pairs of separable FI cubes $\BC^{(n)}_{L_{k+1}}(\Bx)$ and $\BC^{(n)}_{k+1}(\By)$.
\end{theorem}
Above we use the convention $(\textbf{DS}.-1,n,N)$ means no assumption.
\begin{proof}
Consider a pair of separable FI  cubes $\BC^{(n)}_{L_{k+1}}(\Bx)$ and $\BC^{(n)}_{L_{k+1}}(\By)$ and set $J=\kappa(n)+5$. Define
\begin{align*}
\rm{B}_{k+1}&=\left\{\exists E\in I_0: \text{ $\BC^{(n)}_{L_{k+1}}(\Bx)$ and $\BC^{(n)}_{L_{k+1}}(\By)$ are $(E,m)$-S}\right\}\\
\Sigma&=\left\{\exists E\in I_0:\text{ neither $\BC^{(n)}_{L_{k+1}}(\Bx)$ nor $\BC^{(n)}_{L_{k+1}}(\By)$ is $E$-CNR}\right\}\\
\rm{S}_{\Bx}&=\left\{\exists E\in I_0: \text{ $M(\BC^{(n)}_{L_{k+1}}(\Bx);E)\geq J+1$}\right\}\\
\rm{S}_{\By}&=\left\{\exists E\in I_0: \text{ $M(\BC^{(n)}_{L_{k+1}}(\By);E)\geq J+1$}\right\}
\end{align*}
Let $\omega\in\rm{B}_{k+1}$. If $\omega\notin  \Sigma\cup\rm{S}_{\Bx}$, then $\forall E\in I_0$ either $\BC^{(n)}_{L{k+1}}(\Bx)$ or $\BC^{(n)}_{L_{k+1}}(\By)$ is $E$-CNR  and $M(\BC^{(n)}_{L_{k+1}}(\Bx),E)\leq J$. The cube  $\BC^{(n)}_{L_{k+1}}(\Bx)$ cannot be $E$-CNR: indeed, by  Lemma \ref{lem:CNR.NS} it would be $(E,m)$-NS. So the cube $\BC^{(n)}_{L_{k+1}}(\By)$ is $E$-CNR and $(E,m)$-S. This implies again by Lemma \ref{lem:CNR.NS} that
\[
M(\BC^{(n)}_{L_{k+1}}(\By),E)\geq J+1.
\]
Therefore $\omega\in\rm{S}_{\By}$, so that $\rm{B}_{k+1}\subset \Sigma\cup \rm{S}_{\Bx}\cup\rm{S}_{\By}$, hence
\[
\prob\{\rm{B}_{k+1}\}\leq \prob\{\Sigma\}+\prob\{\rm{S}_{\Bx}\}+\prob\{\rm{S}_{\By}\},
\]
and $\prob\{\Sigma\}\leq L_{k+1}^{-4^Np} $ By Theorem \ref{thm:Wegner}. Now let us estimate $\prob\{S_{\Bx}\}$ and similarly $\prob\{\rm{S}_{\By}\}$. Since 
\[
M_{\PI}(\BC^{(n)}_{L_{k+1}}(\Bx),E)+M_{\FI}(\BC^{(n)}_{L_{k+1}}(\Bx),E)\geq M(\BC^{(n)}_{L_{k+1}}(\Bx),E),
\]
the inequality $M(\BC^{(n)}_{L_{k+1}}(\Bx),E)\geq \kappa(n)+6$ implies that either $M_{\PI}(\BC^{(n)}_{L_{k+1}}(\Bx),E)\geq \kappa(n)+2$ or,  $M_{\FI}(\BC^{(n)}_{L_{k+1}}(\Bx),E)\geq 4$. Therefore, by Lemma \ref{lem:prob.M}       and Lemma \ref{lem:prob.M2} with ($\ell=2$),
\begin{align*}
\prob\{\rm{S}_{\Bx}\}&\leq \prob\left\{\exists E\in I: M_{\PI}(\BC^{(n)}_{L_{k+1}}(\Bx),E)\geq \kappa(n)+2\right\}\\
& +\prob\left\{\exists E\in I: M_{\FI}(\BC^{(n)}_{L_{k+1}}(\Bx),E)\geq 4\right\}\\
&\leq \frac{3^{2nd}}{2} L_{k+1}^{2nd}(L_k^{-4^Np}+L_k^{-4^Np4^{N-n }})+C'(n,N,d)L_{k+1}^{4dn-\frac{4p}{\alpha}4^{N-n}}\\
&\leq C''(n,N,d)\left(L_{k+1}^{-\frac{4^Np}{\alpha}+2nd}+L_{k+1}^{-\frac{4p}{\alpha}4^{N-n}+2nd}+L_{k+1}^{-\frac{4p}{\alpha}4^{N-n}+4nd}\right)\\
&\leq C'''(n,N,d)L_{k+1}^{-\frac{4p}{\alpha}4^{N-n}+4nd}\\
&\leq \frac{1}{4} L_{k+1}^{-2p 4^{N-n}},   
\end{align*}
where we used that $\alpha=3/2$, $p\geq 4\alpha Nd=6Nd$. Finally 
\[
\prob\{\rm{B}_{k+1}\}\leq L_{k+1}^{-4^Np}+\frac{1}{2}L_{k+1}^{-2p4^{N-n}}\leq L_{k+1}^{-2p4^{N-n}}.
\]
\end{proof}

\subsection{Mixed pairs of cubes} \label{sec:mixed.cubes}
Finally, it remains only to derive $\dskonN$ in case (III) i.e., for pairs of $n$-particle cubes where one is PI while the other is FI.

\begin{theorem} \label{thm:MI.cubes}
Let $1\leq n\leq N$. There exists $L_3^*=L_3^*(N,d)\geq 0$ such that if $L_0\geq L_3^*(N,d)$ and if for $k\geq 0$
\begin{enumerate}
\item[(i)] 
$\dskunnprimeN$ holds true all $n'\in[1,n)$,\\
\item[(ii)]
$\dsknprimeN$ holds true for all $ n'\in[1,n)$ and \\
\item[(iii)]
$\dsknN$  holds true  for all pairs of FI cubes
\end{enumerate}
then $\dskonN$ holds true for any pair of separable cubes $\BC^{(n)}_{L_{k+1}}(\Bx)$ and $\BC^{(n)}_{L_{k+1}}(\By)$ where one is PI while the other is FI.
\end{theorem}

\begin{proof}
consider a pair of separable $n$-particle cubes $\BC^{(n)}_{L_{k+1}}(\Bx)$, $\BC^{(n)}_{L_{k+1}}(\By)$ and suppose that
 $\BC^{(n)}_{L_{k+1}}(\Bx)$ is PI while $\BC^{(n)}_{L_{k+1}}(\By)$ is FI. Set $J=\kappa(n)+5$ and introduce the events
\begin{align*}
\rm{B}_{k+1}&=\left\{\exists E\in I_0: \text{ $\BC^{(n)}_{L_{k+1}}(\Bx)$ and $\BC^{(n)}_{L_{k+1}}(\By)$ are $(E,m)$-S}\right\}\\
\Sigma&=\left\{\exists E\in I_0: \text{neither $\BC^{(n)}_{L_{k+1}}(\Bx)$ nor $\BC^{(n)}_{L_{k+1}}(\By)$ is $E$-CNR}\right\}\\
\rm{T}_{\Bx}&=\left\{ \text{$\BC^{(n)}_{L_{k+1}}(\Bx)$ is $(E,m)$-T}\right\}\\
\rm{S}_{\By}&=\left\{\text{$\exists E\in I_0$: $M(\BC^{(n)}_{L_{k+1}}(\By),E)\geq J+1$}\right\}   
\end{align*}
Let $\omega\in\rm{B}_{k+1}\setminus(\Sigma\cup\rm{T}_{\Bx})$ then, for all $E\in I_0$ either $\BC^{(n)}_{L_{k+1}}(\Bx)$ is $E$-CNR or $\BC^{(n)}_{L_{k+1}}(\By)$ is $E$-CNR and $\BC^{(n)}_{L_{k+1}}(\Bx)$ is $E,m)$-NT. The cube $\BC^{(n)}_{L_{k+1}}(\Bx)$ cannot be $E$-CNR. Indeed by Lemma \ref{lem:NR.NS} it would have been $(E,m)$-NS. Thus the cube $\BC^{(n)}_{L_{k+1}}(\By)$ is $E$-CNR, so by Lemma \ref{lem:CNR.NS} $M(\BC^{(n)}_{L_{k+1}}(\By),E)\geq J+1$: otherwise $\BC^{(n)}_{L_{k+1}}(\By)$ would be $(E,m)$-NS. Therefore $\omega\in\rm{S}_{\By}$. Consequently,
\[
\rm{B}_{k+1}\subset \Sigma\cup T_{\Bx}\cup\rm{S}_{\By}.
\]
Recall that the probabilities $\prob\{\rm{T}_{\Bx}\}$ and $\prob\{ \rm{S}_{\By}\}$ have already been estimated in Sections \ref{sec:PI.cubes} and \ref{sec:FI.cubes}. We therefore obtain
\begin{align*}
\prob\{\rm{B}_{k+1}\}&\leq \prob\{\rm{T}_{\Bx}\}+\prob\{\rm{S}_{\By}\}\\
&\leq L_{k+1}^{-4^Np} + \frac{1}{2}L_{k+1}^{-4p4^{N-n}}+\frac{1}{4}L_{k+1}^{-2p4^{N-n}}\leq L_{k+1}^{-2p4^{N-n}}
\end{align*}
\end{proof}
\section{Conclusion: The multi-particle multi-scale analysis} \label{sec:MSA}
\begin{theorem}\label{thm:MSA}
Let $1\leq n\leq N$ and $\BH^{(n)}(\omega)=-\BDelta +\sum_{j=1}^n V(x_j,\omega) +\BU$, where $\BU$, $V$ satisfy $\condI$ and $\condP$ respectively. There exists a positive $m$ such that for any $p\geq 6Nd$ property $\dsknN$ holds true for all $k\geq 0$ provided $L_0$ is large enough.
\end{theorem}
\begin{proof}
We prove that for each $n=1,\ldots,N$, property $\dsknN$ is valid. To do so, we use an induction on the number of particles $n'=1,\ldots,n$. For $n=1$ the property holds true for all $k\geq 0$ by the single-particle localization theory \cite{Sto01}. Now suppose that for all $n'\in[1,n)$ $\dsknprimeN$ holds true for all $k\geq 0$, we aim to prove that $\dsknN$ holds true for all $k\geq 0$. For $k=0$, the property is valid using  Theorem \ref{thm:ILS}. Next, suppose that $\dskprimenN$ holds true for all $k'\in(0,k)$, then by combining this last assumption with $\dsknprimeN$ above, one can conclude that:

\begin{enumerate}
\item[(i)]
 $\dsknN$ holds true  for all $k\geq 0$ and for all pairs of PI cubes using Theorem \ref{thm:PI.cubes}\\
\item[(ii)] 
$\dsknN$ holds true for all $k\geq 0$ and for all pairs of FI cubes using Theorem \ref{thm:FI.cubes}\\
\item[(iii)] 
$\dsknN$ holds true for all $k\geq 0$ and for all pairs of MI cubes using Theorem \ref{thm:MI.cubes}
\end{enumerate}

Hence, Theorem \ref{thm:MSA} is proven.
\end{proof} 

\section{Proofs of the results} \label{sec:proof.results}

\subsection{Proof of Theorem \ref{thm:main.result.exp.loc}}

Using the multi-particle multi-scale analysis bounds in the continuum property $\dskNN$, we extend to multi-particle systems the strategy of Stollmann \cite{Sto01}.

For $\Bx_0\in\DZ^{Nd}$ and an integer $k\geq 0$, using the notations of lemma \ref{lem:separable}
\[
R(\Bx_0):= \max_{1\leq \ell\leq \kappa(N)}|\Bx_0-\Bx_{(\ell)}|; \quad b_k(\Bx_0):=7N+ R(\Bx_0)L_k^{-1},
\]

\[
M_k(\Bx_0):=\bigcup_{\ell=1}^{\kappa(N)} C^{(N)}_{7NL_k}(\Bx^{(\ell)})
\]
and define
\[
A_{k+1}(\Bx_0):=\BC^{(N)}_{bb_{k+1}L_{k+1}}(\Bx_0)\setminus\BC^{(N)}_{b_kL_k}(\Bx_0).
\] 
where the positive parameter $b$ is to be chosen later. We can easily check that
\[
M_k(\Bx_0)\subset\BC^{(N)}_{b_kL_k}(\Bx_0).
\]  
Moreover, if $\Bx\in A_{k+1}(\Bx_0)$, then the cubes $\BC^{(N)}_{L_k}(\Bx)$ and $\BC^{(n)}_{L_k}(\Bx_0)$ are separable by Lemma \ref{lem:separable}. Now, also define
\[
\Omega_k(\Bx_0):=\{ \text{$\exists E\in I_0$ and $\Bx\in A_{k+1}(\Bx_0)\cap \Gamma_k:$ $\BC^{(n)}_{L_k}(\Bx)$ and $\BC^{(n)}_{L_k}(\Bx_0)$ are $(E,m)$-S}\},
\]
with $\Gamma_k:=\Bx_0+\frac{L_k}{3}\DZ^{Nd}$. Now property $\dskNN$ combined with the cardinality of $A_{k+1}(\Bx_0)\cap\Gamma_k$ imply 
\begin{align*}
\prob\{\Omega_k(\Bx_0)\}&\leq(2bb_{k+1}L_{k+1})^{Nd}L_k^{-2p} \\
&\leq (2bb_{k+1})^{Nd}L_{k}^{-2p+\alpha Nd}.
\end{align*}
Since, $p\geq (\alpha Nd+1)/2$ (in fact $p\geq 6Nd$), we get $\sum_{k=0}^{\infty} \prob\{\Omega_k(\Bx_0)\}$ is finite. 
Thus, setting
\[
\Omega_{\infty}:=\{ \forall \Bx_0\in\DZ^{Nd}, \text{ $\Omega_k(\Bx_0)$ occurs finitely many times}\},
\]
by the Borel Cantelli Lemma and the countability of $\DZ^{Nd}$ we have that $\prob\{\Omega_{\infty}\}=1$ Therefore it suffices to pick $\omega\in \Omega_{\infty}$ and prove  the exponential decay  of any nonzero eigenfunction $\BPsi$ of $\BH^{(N)}(\omega)$. 

Let $\BPsi$ be a polynomially bounded eigenfunction satisfying (EDI) (see Theorem \ref{thm:EDI}). Let $\Bx_0\in\DZ^{Nd}$ with positive $\|\Bone_{\BC^{(N)}_1(\Bx_0)}\BPsi\| $ (if there is no such $\Bx_0$, we are done.) The cube $\BC^{(N)}_{L_k}(\Bx_0)$ cannot be $(E,m)$-NS for infinitely many $k$. Indeed, given an integer $k\geq 0$, if $ \BC^{(N)}_{L_k}(\Bx_0)$ is $(E,m)$-NS then by (EDI) and the polynomial bound on $\BPsi,$ we get
\begin{align*}
\|\Bone_{\BC^{(N)}_1(\Bx_0)}\BPsi\| \leq&\|\Bone_{\BC^{(N,out)}_{L_k}(\Bx_0)}\BG^{(N)}_{\BC^{(N)}_{L_k}(\Bx_0)}(E)\Bone_{\BC^{(N,int)}_{L_k}(\Bx_0)}\|\cdot\|\Bone_{\BC^{(N,out)}_{L_k}(\Bx_0)}\BPsi\|\\
&\leq C(1+|\Bx_0|+L_k)^t\cdot\ee^{-mL_k}
\end{align*}
and the last term tends to $0$ as $L_k$ tends to infinity in contradiction with the choice of $\Bx_0$. So there is an integer $k_1=k_1(\omega,E,\Bx_0)$ finite such that $\forall k\geq k_1$ the cube $\BC^{(N)}_{L_k}(\Bx_0)$ is $(E,m)$-S. At the same time, since $\omega\in\Omega_{\infty}$, there exists $k_2=k_2(\omega,\Bx_0$ such that if $k\geq k_2$ $\Omega_k(\Bx_0)$ does not occurs. We conclude that for all $k\geq \max\{k_1,k_2\}$, for all $\Bx\in A_{k+1}(\Bx_0)\cap \Gamma_k$, $\BC^{(N)}_{L_k}(\Bx)$ is $(E,m)$-NS.
Let $\rho\in(0,1)$ and choose positive $b$ such that 
\[
b\geq \frac{1+\rho}{1-\rho},
\]
so that 
\[
\tilde{A}_{k+1}:=\BC^{(N)}_{\frac{bb_{k+1}L_{k+1}}{1+\rho}}(\Bx_0)\setminus\BC^{(N)}_{\frac{b_kL_k}{1-\rho}}(\Bx_0)\subset A_{k+1}(\Bx_0),
\]
for $\Bx\in\tilde{A}_{k+1}(\Bx_0)$.  
\begin{enumerate}
\item[(1)]
Since, $|\Bx-\Bx_0|\geq\frac{b_kL_k}{1-\rho}$,
\begin{align*}
\dist(\Bx,\partial\BC^{(N)}_{b_kL_k}(\Bx_0)&\geq |\Bx-\Bx_0|-b_kL_k\\
&\geq|\Bx-\Bx_0|-(1-\rho)|\Bx-\Bx_0|\\
&=\rho|\Bx-\Bx_0|
\end{align*}
\item[(2)]
Since $|\Bx-\Bx_0|\leq \frac{bb_{k+1}L_{k+1}}{1+\rho}$,
\begin{align*}
\dist(\Bx,\partial\BC^{(N)}_{bb_{k+1}L_{k+1}}(\Bx_0))&\geq bb_{k+1}L_{k+1}-|\Bx-\Bx_0|\\
&\geq (1+\rho)|\Bx-\Bx_0| - |\Bx-\Bx_0|\\
&=\rho|\Bx-\Bx_0|.
\end{align*}
\end{enumerate} 
Thus, 
\[
\dist(\Bx,\partial A_{k+1}(\Bx_0))\geq \rho|\Bx-\Bx_0|.
\]
Now, setting $k_3=\max\{k_1,k_2\}$, the assumption linking $b$ and $\rho$ implies that
\[
\bigcup_{k\geq k_3}\tilde{A}_{k+1}(\Bx_0)=\DR^{Nd}\setminus\BC^{(N)}_{\frac{b_{k_3}L_{k_3}}{1-\rho}}(\Bx_0).
\]
because $\frac{bb_{k+1}L_{k+1}}{1+\rho}\geq \frac{b_kL_k}{1-\rho}$. Let $k\geq k_3$, recall that this implies that all the cubes with centers in $A_{k+1}(\Bx_0)\cap\Gamma_k$ and side length $2L_k$ are $(E,m)$-NS. Thus, for any $\Bx\in\tilde{A}_{k+1}(\Bx_0)$, we choose $\Bx_1\in A_{k+1}(\Bx_0)$ such that $\Bx\in\BC^{(n)}_{L_k}(\Bx_1)$. Therefore
\begin{align*}
\|\BC_1^{(N)}(\Bx)\BPsi\|&\leq \|\Bone_{\BC^{(N,int)}_{L_k}(\Bx_1)}\BPsi\|\\
&\leq C\cdot \ee^{-mL_k}\|\cdot\|\Bone_{\BC^{(N,out)}_{L_k}(\Bx_1)}\BPsi\|.
\end{align*}
Up to a set of Lebesgue measure zero, we can cover $\BC^{(N,out)}_{L_k}(\Bx_1)$ by at most $3^{Nd}$ cubes  
\[
\BC^{(N,int)}_{L_k}(\tilde{\Bx}),\quad \tilde{\Bx}\in\Gamma_k,\quad |\tilde{\Bx}-\Bx_1|=\frac{L_k}{3}.
\]
By choosing $\Bx_2$ which gives a maximal norm, we get
\[
\|\Bone_{\BC^{(N,out)}_{L_k}(\Bx_1)}\BPsi\|\leq 3^{Nd}\cdot\|\Bone_{\BC_{L_k}^{(N,int)}(\Bx_2)}\BPsi\|,
\]
so that
\[
\|\Bone_{\BC^{(N)}_1(\Bx)}\BPsi\|\leq 3^{Nd}\cdot\ee^{-mL_k}\cdot\|\Bone_{\BC^{(N,int)}_{L_k}(\Bx_2)}\BPsi\|.
\]

Thus, by an induction procedure, we find a sequence $\Bx_1,\Bx_2,\ldots,\Bx_n$ in $\Gamma_k\cap A_{k+1}(\Bx_0)$ with the bound

\[
\|\Bone_{\BC^{(N)}_1(\Bx)}\BPsi\| \leq (C\cdot 3^{Nd}\exp(-mL_k))^n\cdot\|\Bone_{\BC^{(N,out)}_{L_k}(\Bx_n)}\BPsi\|.
\]
Since $|\Bx_i-\Bx_{i+1}|=L_k/3$ and $\dist(\Bx,\partial A_{k+1})\geq \rho\cdot|\Bx-\Bx_0|$, we can iterate at least $\rho\cdot|\Bx-\Bx_0|\cdot 3/{L_k}$ times until, we reach the boundary of $A_{k+1}(\Bx_0)$. Next, using the polynomial bound on $\BPsi$, we obtain:

\begin{align*}
\|\Bone_{\BC^{(N)}_1(\Bx)}\BPsi\|&\leq (C\cdot 3^{Nd})^{\frac{3\rho|\Bx-\Bx_0|}{L_k}}\cdot\exp(-3m\rho|\Bx-\Bx_0|)\\
&\times C(1+|\Bx_0|+bL_{k+1})^t\cdot L_{k+1}^{Nd}.
\end{align*}
We can conclude that given $\rho'$ with $\rho'\in(0,1)$, we can find $ k_4\geq k_3$ such that if $k\geq k_4$, then 
\[
\|\Bone_{\BC_1^{(N)}(\Bx)}\BPsi\|\leq \ee^{-\rho\rho'm|\Bx-\Bx_0|},
\]
if $|\Bx-\Bx_0|\geq \frac{b_{k_4}L_{k_4}}{1-\rho}$. This completes the proof of the exponential localization in the max-norm.               
\subsection{Proof of Theorem \ref{thm:main.result.dynamical.loc}}
For the proof of the multi-particle dynamical localization given the multi-particle multi-scale analysis in the continuum, we refer to the paper by Boutet de Monvel et al. \cite{BCS11}.

\section{Appendix}\label{sec:appendix}

\subsection{proof of Lemma \ref{lem:separable}}
(A) Consider positive $L$, $\emptyset\neq \CJ\subset\{1,\ldots,n\}$ and $\By\in\DZ^{nd}$. $\{y_j\}_{j\in\CJ}$ is  called an $L$-cluster if the union 
\[
\bigcup_{j\in\CJ} C^{(1)}_L(y_j),
\]
cannot be decomposed into two non-empty disjoint subsets. Next, given two configurations $\Bx,\By\in\DZ^{nd}$, we proceed as follows:
\begin{enumerate}
\item[(1)]  We decompose the vector $\By$ into  maximal $L$-clusters $\Gamma_1,\ldots=,\Gamma_M$ (each of diameter $\leq 2nL$) with $M\leq n$\\
\item[(2)]
Each position $y_i$ corresponds to exactly one cluster $\Gamma_j,$ $j=j(i)\in\{1,\ldots,M\}.$\\
\item[(3)]
If there exists $j\in\{1,\ldots,M\}$ such that $\Gamma_j\cap \varPi\BC^{(n)}_{L_k}(\Bx)=\emptyset$, then the cubes $\BC^{(n)}_{L_k}(\By)$ and $\BC^{(n)}_{L_k}(\Bx)$ are separable\\
\item[(4)]
If (3) is wrong, then for all $k=1,\ldots,M$ $\Gamma_k\cap\varPi\BC^{(n)}_{L}(\Bx)\neq \emptyset$. Thus for all $k=1,\ldots,M$, $\exists i=1,\ldots,n$ such that $\Gamma_k\cap C^{(1)}_L(x_i)\neq \emptyset$. Now for any $j=1,\ldots,n$ there exists $k=1,\ldots,M$ such $y_j\in\Gamma_k$. Therefore for such $k$, by hypothesis there exists $i=1,\ldots,n$ such that $\gamma_k\cap C^{(1)}_L(x_i)\neq \emptyset$. Next let $z\in\Gamma_k\cap C^{(1)}_L(x_i)$ so that $|z-x_i|\leq L$. We have that 
\begin{align*}
|y_j-x_i|&\leq |y_j-z|+|z-x_i|\\
&\leq 2nL-L+L=2nL
\end{align*}
since $y_j\in\Gamma_k$. 
\end{enumerate}

Notice that above we have the bound $|y_j-z|\leq 2nL-L$ because $y_j$is a center of the $L$-cluster $\Gamma_k$ Hence for all $j=1,\ldots,n$ $y_j$ must belong to one of the cubes $C^{(1)}_{2nL}(x_i)$ for the $n$-positions $(y_1,\ldots,y_n)$. Set $\kappa(n)=n^n$. For any choice of at most $\kappa(n)$ possibilities; $\By=(y_1,\ldots,y_n)$ must belong to the cartesian product of $n$ cubes of side length $2L$ i.e., an $nd$-dimensional cube of size $2nL$, the assertion then follows.

(B) Set $R(\By)=\max_{1\leq i,j\leq n}|y_i-y_j| + 5NL$ and consider a cube $\BC^{(n)}_L(\Bx)$  with $|\By-\Bx|\geq R(\By)$. Then there exist $i_0\in\{1,\ldots,n\}$ such that $|y_{i_0}-x_{i_0}|\geq R(\By)$. Consider  the maximal connected component $\Lambda_{\Bx}:=\bigcup_{i\in\CJ} C^{(1)}_L(x_i)$ of the union $\bigcup_i C^{(1)}_L(x_i)$ containing $x_{i_0}$. Its diameter bis bounded by $2nL$. We have 

\[
\dist(\Lambda_{\Bx};\varPi\BC^{(n)}_L(\By))=\min_{u,v}|u-v|,
\]
now, since
\[
|x_{i_0}-y_{i_0}|\leq |x_{i_0}-u|+|u-v|+|v-y_{i_0}|,
\]
then
\begin{align*}
\dist(\Lambda_\Bx,\varPi\BC^{(n)}_L(\By))&=\min_{u,v}|u-v |-\diam(\Lambda_\Bx)-\max_{v,y_{i_0}}|v-y_{i_0}|.\\
\end{align*} 
Recall that $\diam(\Lambda_{\Bx})\leq 2nL$ and 
\[
\max_{v,y_{i_0}}|v-y_{i_0}|\leq \max_{v}|v-y_j|+\max_{y_{i_0}}|y_j-y_{i_0}|,
\] for some $j=1,\ldots,n$ such that $v\in C^{(1)}_L(y_j)$. Finally, we get 
\[
\dist(\Lambda_{\Bx},\varPi\BC^{(n)}_L(\By))\geq R(\By)-\diam(\Lambda_{\Bx})-(2L+\diam(\varPi\By)),
\]
and the latter quantity is strictly positive. This implies that $\BC^{(n)}_L(\Bx)$ is $\CJ$separable from $\BC^{(n)}_L(\By)$.

\subsection{Proof of Lemma \ref{lem:PI}}
Set $R:=2L+r_0$ and assume that $\diam\varPi\Bu=\max_{i,j}|u_i-u_j|\geq nR$. If the union of cubes $C^{(1)}_{R/2}(u_i)$, $i=1,\ldots,n$ were not decomposable into two (or more) disjoint groups, then, it would be connected hence its diameter would be bounded by $n(2(R/2))=nR$ hence $\diam\varPi\Bu\leq nR$ which contradicts the hypothesis. Therefore, there exists an index subset $\CJ\subset\{1,\ldots,n\}$ such that  $|u_{j_1}-u_{j_2}|\geq 2(R/2)$ for all $j_1\in\CJ$ and $j_2\in\CJ^c$, this implies that 

\begin{align*}
\dist\left(\varPi_{\CJ}\BC^{(n)}_L(\Bu),\varPi_{\CJ^c}\BC^{(n)}_L(\Bu)\right)&=\min_{j_1\in\CJ,j_2\in\CJ^c}\dist\left( C^{(1)}_{L}(u_{j_1}),C^{(1)}_L(u_{j_2})\right)\\
&\geq \min_{j_1\in\CJ,j_2\in\CJ^c}|u_{j_1}-u_{j_2}|-2L\geq r_0.
\end{align*}

\subsection{Proof of Lemma \ref{lem:FI.cubes}} 
If for some positive $R$
\[
R\leq|\Bx-\By|=\max_{1\leq j\leq n}|x_j-y_j|,
\]
then there exists $1\leq j_0\leq n$ such that $|x_{j_0}-y_{j_0}|\geq R$. Since both cubes are fully interactive,
\begin{align*}
&|x_{j_0}-x_i |\leq \diam \varPi_{\Bx}\leq n(2L+r_0),\\
&|y_{j_0}-y_j|\leq \diam\varPi_{\By} \leq n(2L+r_0).\\
\end{align*}

By the triangle inequality, for any $1\leq i,j\leq n$ and $R\geq 7nL\geq 6nL+2nr_0$, we have 
\begin{align*}
|x_i-y_j|&\geq |x_{j_0}-y_{j_0}|-|x_{j_0}-x_i|-|y_{j_0}-y_j|\\
&\geq 6nL+2nr_0-2n(2L+r_0)=2nL.\\
\end{align*}
Therefore, for any $1\leq i,j\leq n$,
\[
\min_{i,j}\dist\left( C^{(1)}_L(x_i),C^{(1)}_L(y_j)\right)\geq \min_{i,j}|x_i-y_j|-2L\geq 2(n-1)L.
\]
which proves the claim. 

\subsection{Proof of Lemma \ref{lem:M}}
Assume that $M^{\sep}(\BC^{(n)}_{L_{k+1}}(\Bu),E)$ is less than $2$ (i.e.,there is no pair of separable cubes of radius $L_k$ in $\BC^{(n)}_{L_{k+1}}(\Bu))$ but $M(\BC^{(n)}(\Bu),E)\geq \kappa(n)+2$. Then $\BC^{(n)}_{L_{k+1}}(\Bu)$  must contain at least $\kappa(n)+2$ cubes $\BC^{(n)}_{L_k}(\Bv_i)$, $0\leq i \leq \kappa(n)+1$ which are not separable but satisfy $|\Bv_i-\Bv_{i'}|\geq 7NL_k$ for all $i\neq i'$. On the other hand, by Lemma \ref{lem:separable} there are at most $\kappa(n)$ cubes   $\BC^{(n)}_{2nL_k}(\By_i)$, such that any cube $\BC^{(n)}_{L_k}(\Bx)$ with $\Bx\notin \bigcup_j \BC^{(n)}_{2nL_k}(\By_j)$, is separable from $\BC^{(n)}_{L_k}(\Bv_0)$. Hence $\Bv_i\in\bigcup_j\BC^{(n)}_{2nL_k}(\By_j)$ for all $i=1,\ldots,\kappa(n)+1$. But since for all $i\neq i'$ $|\Bv_i-\Bv_{i'}|\geq 7NL_k$ there must be at most one center $\Bv_i$ per cube $\BC^{(n)}_{2nL_k}(y_j)$, $1\leq j\leq \kappa(n)$. Hence we come to a contradiction 
\[
\kappa(n)+1\leq \kappa(n).
\]
The same analysis holds true if we consider only PI cubes.

\subsection{Proof of Lemma \ref{lem:prob.M}}
Suppose that $M_{\PI}(\BC^{(n)}_{L_{k+1}}(\Bu),I)\geq \kappa(n)+2$, then by Lemma \ref{lem:M} $M^{\sep}_{\PI}(\BC^{(n)}_{L_{k+1}}(\Bu),I)\geq 2$ i.e., there are at least two separable $(E,m)$-S PI cubes $\BC^{(n)}_{L_k}(\Bu^{(j_1)}$, $\BC^{(n)}_{L_k}(\Bu^{(j_2)})$  inside $\BC^{(n)}_{L_{k+1}}(\Bu)$. The number of possible pairs of centers  $\{\Bu^{(j_1)},\Bu^{(j_2)}\}$ such that 
\[
\BC^{(n)}_{L_k}(\Bu^{(j_1)}), \BC^{(n)}_{L_k}(\Bu^{(j_2)})\subset \BC^{(n)}_{L_{k+1}}(\Bu)
\]
is bounded by $\frac{3^{2nd}}{2}L_{k+1}^{2nd}$. Then, setting
\[
\rm{B}_k=\{\exists E\in I, \BC^{(n)}_{L_k}(\Bu^{(j_1)}), \BC^{(n)}_{L_k}(\Bu^{(j_2)}) \text{ are $(E,m)$-S}\}
\]
\[
\prob\left\{M^{\sep}_{\PI}(\BC^{(n)}_{L_{k+1}}(\Bu),I)\geq 2\right\}\leq \frac{3^{2nd}}{2}L_{k+1}^{2nd}\times \prob\{B_k\}
\]
with $\prob\{\rm{B}_k\}\leq L_k^{-4^Np} + L_k^{-4p4^{N-n}}$  

\subsection{Proof of Lemma \ref{lem:prob.M2}}
Suppose there exist $2\ell$ pairwise separable fully interactive cubes $\BC^{(n)}_{L_k}(\Bu^{(j)})\subset\BC^{(n)}_{L_{k+1}}(\Bu)$, $1\leq j\leq 2\ell.$ Then  by Lemma \ref{lem:FI.cubes} for any pair $\BC^{(n)}_{L_k}(\Bu^{(2i-1)})$, $\BC^{(n)}_{L_k}(\Bu^{(2i)})$  the corresponding random Hamiltonians $\BH^{(n)}_{\BC^{(n)}_{L_k}(\Bu^{(2i-1)})}$ and $\BH^{(n)}_{\BC^{(n)}_{L_k}(\Bu^{(2i)})}$ are independent and so are their spectra and their Green functions. For $i=1,\ldots,\ell,$ we consider the events:
\[
\rm{A}_i=\left\{ \exists E\in I: \BC^{(n)}_{L_k}(\Bu^{(2i-1)}) \text{ and $\BC^{(n)}_{L_k}$ are } (E,m)-S\right\}.
\]
then by assumption  $\dsknN$, we have for $i=1,\ldots,\ell$
\[
\prob\{\rm{A}_i\}\leq L_k^{-2p4^{N-n}}
\]
and by independence of the events $A_1,\ldots,A_{\ell}$
\[
\prob\left\{ \bigcap_{1\leq i\leq \ell}A_i \right\}=\prod_{i=1}^{\ell}\prob\{\rm{A_i}\}\leq \left(L_k^{-2p4^{N-n}}\right)^{\ell}.
\]
To complete the proof, note that the total number of different  families of $2\ell$ cubes  $\BC^{(n)}_{L_k}(\Bu^{(j)})\subset\BC^{(n)}_{L_{k+1}}(\Bu)$, $1\leq j\leq 2\ell$ is bounded  by 
\[
\frac{1}{(2\ell)!}\left|\BC^{(n)}_{L_{k+1}}(\Bu)\right|^{2\ell}\leq C(n,N,d,\ell) L_k^{2\ell dn\alpha}
\]

\bibliographystyle{plain}

\begin{bibdiv}
\begin{biblist}

\bib{AW09}{article}{
    author={M. Aizenmann},
		author={S. Warzel},
		title={Localization bounds for multi-particle systems},
		journal={Commun. Math. Phys.},
		volume={290},
		pages={903--934},
		date={2009}
}

\bib{BCSS10a}{misc}{
 author={A. Boutet de Monvel},
 author={V. Chulaevsky},
 author={P. Stollmann},
  author={Y. Suhov},
	title={Anderson localization  for a multi-particle model with an alloy-type external random potential},
	status={arXiv:math-Ph/1004.1300v1},
	date={2010}
}

\bib{BCSS10b}{article}{
    author={A. Boutet de Monvel},
		author={V. Chulaevsky},
		author={P. Stollmann},
		author={Y. Suhov},
		title={Wegner type-bounds for a multi-particle continuous Anderson model with an alloy-type  external potential},
		journal={J. Stat. Phys.},
		 volume={138},
		pages={553--566},
		date={2010}
}
\bib{BCS11}{article}{
    author={A. Boutet de monvel},
		author={V. Chulaevsky},
		author={Y. Suhov},
		title={Dynamical localization for multi-particle models with an alloy-type external random          potentials},
		journal={nonlinearity},
		volume={24},
		pages={1451--1472},
		date={2011}
}

\bib{CS09}{article}{
   author={V. Chulaevsky},
	 author={Y. Suhov},
	title={Multi-particle Anderson localization. Induction on the number of particles},
	journal={Math. Phys. Anal. Geom.},
	volume={12},
	pages={117--139},
	date= {2009},
}
\bib{C83}{article}{
   author={R. Carmona},
	 title={One-dimensional  Schr\"odinger operators  with random  or deterministics potentials, new spectral types},
	 journal={J. Funct. Anal.},
	volume={51},
	 pages={229--258},
	 date={1983},
}

\bib{DK89}{article}{
    author={H. von Dreifus},
		author={A. Klein},
		title={A new proof of localization  in the Anderson tight-binding model},
		journal={Commun. Math. Phys.},
		volume={124},
		pages={285--299},
		date={1989}
}
\bib{CH94}{article}{
   author={J. M. Combes},
	 author={P. D. Hislop},
	 title={Localization for continuous  random Hamiltonians in $d$-dimensions},
	 journal={J. Funct. Anal.},
	 volume={124},
	 pages={149--180}
}
	 
\bib{CL90}{book}{
     author={R. Carmona},
		 author={J. Lacroix},
		 title={Spectral theory of random Schr\"odinger operators},
		publisher={Birkh\"auser Boston},
		place={Boston Inc.},
		date={1990}
}
\bib{CS08}{article}{
    author={V. Chulaevsky},
		author={Y. Suhov},
		title={Wegner bounds for a two particle tight-binding model},
		journal={Commun. Math. Phys.},
		volume={283},
		pages={479--489},
		date={2008}
}
\bib{DS01}{article}{
  author={D. Damanik},
	author={P. Stollmann},
	title={Multi-scale analysis implies dynamical localization},
	journal={Geom. Funct. Anal.},
	volume={11},
	pages={11--29},
	number={1},
	date={2001}
	}
\bib{DSS02}{article}{
     author={D. Damanik},
		 author={R. Sims},
		 author={G. Stolz},
		 title={Localization fof one-dimensional continuum Bernoulli-Anderson models},
		 journal={Duke Math. Journal},
		 volume={114},
		 pages={59--100},
		 date={2002}
}	

\bib{Eka11}{article}{
    author={T. Ekanga},
		title={On two-particle Anderson localization at low energies},
		journal={C. R. Acad. Sci. Paris Ser I},
		volume={349},
		pages={167--170},
		date={2011}
}
\bib{Eka16}{article}{
    author={T. Ekanga},
		title={Multi-particle localization   for weakly interacting Anderson tight-binding models},
		journal={J. Math. Phys.},
		volume={58},
		date={2016}
}
\bib{Eka19a}{article}{
    author={T. Ekanga},
		title={Localization at low energy in the multi-particle tight-binding Anderson model},
		journal={Rev. Math. Phys.},
		date={2019}
 }
\bib{GK02}{article}{
     author={F. Germinet},
		 author={A. KLein},
		  title={Operator kernel estimates for functions of generalized Schr\"odinger operators},
			journal={Proceeding of the American Mathematical society},
			volume={131},
			pages={911--920},
			date={2002}
}
\bib{Sto01}{book}{
   author={P. Stollmann},
	 title={Caught by disorder bounded states in random media},
	 publisher={Birkh\"auser boston Inc.},
	 place={Boston, MA},
	 date={2001}
}
\bib{KS87}{article}{
   author={S. Kotani},
	 author={B. Simon},
	 title={Localization in general one-dimensional random systems},
	 journal={Commun. Math. Phys.},
	 volume={112},
	 pages={103--119},
	 date={1987}
}
\bib{PF92}{book}{
   author={L. Pastur},
	 author={A. Fogotin},
	 title={Spectra of random and almost-periodic operators},
	 publisher={Springer-Verlag},
	 date={1992}
	}
	\bib{SW86}{article}{
	 author={B. Simon},
	 author={T. Wolf},
	 title={Singular continuous under rank one  perturbation  and localization  for random Hamiltonians },
	 journal={Commun. Pure Appl. Math},
	 volume={39},
	 pages={75--90},
	 date={1986}
}
\bib{S55}{article}{
   author={G. Stolz},
	 title={Localization for the poisson model, " in spectral analysis and partial differential equations" operator  theory: Advances and Applications },
	journal={Birkh\"auser-Verlag},
	volume={78},
	pages={375--380},
	date={1955}
	}
\bib{S95}{article}{
	author={G. Stolz},
	title={Localization for random Schr\"odinger operators with poisson potential},
	journal={Annales de l' I.H.P., Section A},
	volume={63}
	number={3},
	pages={297--314},
	date={1995}
}
	 
\end{biblist}
\end{bibdiv}

\end{document}